\documentclass[10pt,conference]{IEEEtran}
 \usepackage[frozencache,cachedir=.]{minted} 

\ifCLASSOPTIONcompsoc
  \usepackage[nocompress]{cite}
\else
  \usepackage{cite}
\fi
\usepackage[colorlinks]{hyperref}

\ifCLASSINFOpdf
\else
\fi
\usepackage{amsmath}
\usepackage{amsfonts}
\usepackage{amssymb}
\usepackage{stmaryrd}
\usepackage{stackrel}
\usepackage[sort, numbers]{natbib}

\usepackage{todonotes}
\usepackage{graphicx}
\usepackage{amsthm}
\usepackage{mathpartir}
 \usepackage{float}
 \usepackage[all]{xy}

\newcommand{\posreal}{\RR_{>0}}
\newcommand{\negreal}{\RR_{<0}}
\newcommand{\forget}[1]{\lfloor #1 \rfloor}
\newcommand{\score}{\colforsyntax{score}}
\newcommand{\sample}{\colforsyntax{sample}}
\newcommand{\bigstep}{\Downarrow}
\newcommand{\true}{\colforsyntax{true}}
\newcommand{\false}{\colforsyntax{false}}
\newcommand\reals{\colforsyntax{real}}

\newcommand{\pullbackcorner}[1][dr]{\save*!/#1-1.2pc/#1:(-1,1)@^{|-}\restore}
\newcommand{\BRel}{\mathbf{BRel}}
\newcommand{\lcstruct}[0]{\lambda_c(\Sigma)}
\newcommand{\lifted}[0]{\stackrel{\cdot}{\to}}
\newcommand{\Gl}[0]{\mathbf{Gl}}
\newcommand{\ter}[0]{t}
\newcommand{\wcpocat}{\mathbf{\omega CPO}}

\newcommand{\wpap}[0]{\ensuremath{\omega}\text{PAP}}
\newcommand{\sem}[1]{\llbracket #1\rrbracket}
\newcommand{\defeq}{\stackrel {\mathrm{def}}=}
\newcommand{\lift}[0]{\mathbf{L}}

\newcommand{\RR}{\mathbb{R}}
\newcommand{\BB}{\mathbb{B}}

\newcommand{\NN}{\mathbb{N}}

\newcommand{\pap}[0]{PAP}

\newcommand\To{\to}
\newcommand{\var}[0]{x}
\newcommand\tPair[2]{\langle #1, #2\rangle}

\newcommand\tTuple[1]{\langle #1\rangle}
\newcommand\fun[1]{\lambda #1.}

\newcommand\pMatch[5][\,]{\llet\,#1\tPair{#3}{#4}\,=\,#2\,\iin \,#5}
\newcommand\tMatch[4][\,]{\llet\,#1\tTuple{#3}\,=\,#2\,\iin \, #4}

\newcommand\cnst{\underline{c}}

\newcommand{\conc}[1]{|#1|}
\newcommand{\wcpo}[0]{\ensuremath{\omega}cpo}

\newcommand{\site}[0]{(\catC,\topo)}
\newcommand{\concsh}[0]{Conc(\catC,\topo)}
\newcommand{\plots}[1]{\mathcal{P}_{#1}}
\newcommand{\monos}[0]{\mathcal{M}}
\newcommand{\mono}[0]{\rightarrowtail}
\newcommand{\bind}[0]{\mathop{\gg\mkern-10mu\scalebox{1}[1]{=}}}

\newtheorem{theorem}{Theorem}[section]
\newtheorem{corollary}[theorem]{Corollary}
\newtheorem{lemma}[theorem]{Lemma}
\newtheorem{proposition}[theorem]{Proposition}
\newtheorem{definition}{Definition}[section]
\newtheorem{example}{Example}[section]

\newtheorem{remark}{Remark}[section]

\newcommand{\ad}[0]{\mathcal{D}}

\newcommand{\syntaxcolor}[0]{\color{blue!70!black}}

\newcommand{\colforsyntax}[1]{\mathbf{\syntaxcolor{#1}}}

\newcommand{\llet}[0]{\colforsyntax{let}}
\newcommand{\iin}[0]{\colforsyntax{in}}
\newcommand{\iif}[0]{\colforsyntax{if}}
\newcommand{\then}[0]{\colforsyntax{then}}
\newcommand{\eelse}[0]{\colforsyntax{else}}
\newcommand{\iifthenelse}[3]{\iif~#1~\then~#2~\eelse~#3}
\newcommand{\return}[0]{\colforsyntax{return}}

\newcommand{\dom}[0]{Dom}

\newcommand{\catC}{\mathbf{C}}

\newcommand{\Set}[0]{\mathbf{Set}}
\newcommand{\Pred}[0]{\mathbf{Pred}}
\newcommand{\topo}[0]{\mathcal{J}}
\newcommand{\Cartsp}{\mathbf{CartSp}}

\newcommand{\hausd}[0]{\mathcal{H}}

\newcommand{\norm}[1]{\lVert #1 \rVert }
\newcommand{\grad}{\nabla} 
\DeclareMathOperator{\diam}{\mathrm{diam}}

\newcommand{\OpenCont}{\mathbf{OpenCont}}
\newcommand{\ini}[0]{0}
\newcommand{\term}[0]{1}

\newcommand{\Sbs}{\mathbf{Sbs}}
\newcommand{\Injpap}{\mathbf{Inj}}
\newcommand{\Shinj}{\mathbf{Sh(Inj)}}
\newcommand{\Subshinj}{\mathbf{Sub(Sh(Inj))}}
\newcommand{\Cont}{\mathbf{Cont}}
\newcommand{\wconcsh}[0]{\omega\concsh}
\newcommand{\submono}[0]{Sub_\monos}
\newcommand{\cinf}[0]{c_\infty}
\newcommand{\sitewmono}[0]{(\catC,\topo, \monos)}
\newcommand{\cpap}[0]{c\pap}
\newcommand{\monopap}[0]{\monos_{\pap}}
\newcommand{\iso}[0]{\cong}
\newcommand{\topopap}[0]{\topo_{\pap}}
\newcommand{\fib}[3]{#1:\mathbb{#2}\to\mathbb{#3}}
\newcommand{\mbb}{\mathbb{B}}
\newcommand{\mba}{\mathbb{A}}

\newcommand{\mbe}{\mathbb{E}}
\newcommand{\mbw}{\mathbb{W}}

\makeatletter
\newcommand\@TyAlph[1]{%
\ifcase #1\or \tau\or \sigma\or \rho\else \@ctrerr \fi%
}
\newcommand\ty[1][1]{{\@TyAlph{#1}}}

\newcommand\tvar[1][1]{{\@TyVarAlph{#1}}}
\newcommand\@TyVarAlph[1]{%
\ifcase #1\or \alpha\or \beta\or \gamma\else \@ctrerr \fi%
}

\newcommand\@VarAlph[1]{%
\ifcase #1\or x\or y\or z\or u\or v\or w\else \@ctrerr \fi%
}

\newcommand\trm[1][1]{{\@TermAlph{#1}}}
\newcommand\@TermAlph[1]{%
\ifcase #1\or t\or s\or r\else \@ctrerr \fi%
}

\newcommand\syncat[1]{\mspace{-25mu}\synname{#1}}
\newcommand\synname[1]{\quad\text{#1}}

\newenvironment{syntax}[1][]{%
\(%
  \begin{array}[t]{#1l@{\quad\!\!}*3{l@{}}@{\,}l}
}{
\end{array}
\)%
}

\newcommand\gdefinedby{::=}
\newcommand\gor{\mathrel{\lvert}}

\begin{document}
%
\title{$\omega$PAP Spaces: Reasoning Denotationally About Higher-Order, Recursive Probabilistic and Differentiable Programs}
%
%
%
%

 \author{Mathieu~Huot$^*$,~
        Alexander~K.~Lew$^*$,~
         Vikash~K.~Mansinghka,~%
         and~Sam~Staton 
}
\IEEEtitleabstractindextext{%
\begin{abstract}
We introduce a new setting, the category of $\omega$PAP spaces, for reasoning denotationally about expressive differentiable and probabilistic programming languages. Our semantics is \textit{general} enough to assign meanings to 
most practical probabilistic and differentiable programs, including
those that use general recursion, higher-order functions, discontinuous primitives,
and discrete and continuous sampling. But crucially, 
it is also \textit{specific} enough to \textit{exclude} many 
pathological denotations, enabling us to establish new results
about differentiable and probabilistic programs.
In the differentiable setting, we prove general correctness theorems for automatic differentiation and its use within gradient descent. In the probabilistic setting, we establish the almost-everywhere differentiability of probabilistic programs' trace density functions, and the existence of convenient base measures for density computation in Monte Carlo inference. In some cases these results were previously known, but required detailed proofs of an operational flavor; by contrast, all our proofs work directly with programs' denotations.
\end{abstract}

}
\IEEEoverridecommandlockouts \IEEEpubid{\makebox[\columnwidth]{979-8-3503-3587-3/23/\$31.00~ \copyright2023 IEEE \hfill} \hspace{\columnsep}\makebox[\columnwidth]{ }}

\maketitle

\IEEEdisplaynontitleabstractindextext

%
\IEEEpeerreviewmaketitle


%
%
%
%
 \renewcommand*{\thefootnote}{\fnsymbol{footnote}}
\footnotetext[1]{ Equal contribution}
\renewcommand*{\thefootnote}{\arabic{footnote}}
\setcounter{footnote}{0}
\section{Introduction}\label{sec:intro}

This paper introduces a new setting, the category of $\omega$PAP spaces, for reasoning denotationally about expressive probabilistic and differentiable programs, and demonstrates its utility in several applications. The $\omega$PAP spaces are built using the same categorical machinery~\citep{matache2022concrete} that underlies the $\omega$-quasi-Borel spaces~\cite{vakar2019domain} and the $\omega$-diffeological-spaces~\cite{huot2020correctness,vakar2020denotational}, two other recent proposals for understanding higher-order, recursive probabilistic and differentiable programs. 
The key difference is that instead of taking the \textit{measurable} maps (as in $\omega$Qbs) or the \textit{smooth} maps (as in $\omega$Diff) as the primitives in our development, we instead use the functions that are \textit{piecewise analytic under analytic partition}, or PAP~\cite{lee2020correctness}. 

Whereas the smooth functions \textit{exclude} many primitives used in practice (e.g., $< \, : \reals \to \reals \to \mathbb{B}$), and the measurable functions \textit{admit} many pathological examples from analysis (the Cantor function, the Weierstrass function, space-filling curves), the PAP functions manage to exhibit very few pathologies while still including nearly all primitives exposed by today's differentiable and probabilistic programming languages. As a result, our semantics can interpret most differentiable and probabilistic programs that arise in practice, \textit{and} can be used to provide short denotational proofs of many interesting properties. As evidence for this claim, we use our semantics to establish new results (and give new, simplified proofs of old results) about the correctness of automatic differentiation, the convergence of automatic-differentiation gradient descent, and the supports and densities of recursive probabilistic programs. 

\subsection{A ``Just-Specific-Enough'' Semantic Model}

In denotational accounts of deterministic and probabilistic programming languages, deterministic programs typically denote \textit{functions} and probabilistic programs typically denote \textit{measure kernels}. But what kinds of functions and kernels? 

The more specific our answer to that question, the more we can hope to prove\textemdash purely denotationally\textemdash about our languages. For example, consider the question of \textit{commutativity}: is it the case that, whenever $x$ does not occur free in $s$, $$\sem{\llet~x=t~\iin~\llet~y=s~\iin~u} = \sem{\llet~y=s~\iin~\llet~x=t~\iin~u}?$$
In the probabilistic setting, this amounts to asking whether iterated integrals can always be reordered. If our semantics interprets programs as \textit{arbitrary} measure kernels, then there is no obvious answer, because Fubini's theorem, which justifies the interchange of iterated integrals, does not apply unconditionally. But by interpreting an expressive probabilistic language using only the \textit{s-finite kernels}, \citet{staton2017commutative} was able to prove commutativity denotationally, using the specialization of Fubini's theorem to the s-finite case.

Unfortunately, in the quest for semantic models with convenient properties, we may end up excluding programs that programmers might write in practice. For example, in~\cite{huot2020correctness}, all programs are interpreted as \textit{smooth} functions, enabling an elegant semantic account of why automatic 
differentiation algorithms work (even in the presence of higher-order functions). But to achieve this nice theory, all non-smooth
\textit{primitives} (including, e.g., the function $\lambda x. \lambda y. x < y$) were excluded from the language under consideration. As a result, the theory has little to say about the behavior of automatic differentiation on non-smooth programs.

In this work, our aim is to find a ``just specific enough'' semantic model of expressive probabilistic and differentiable programming languages: specific enough to establish interesting properties via denotational reasoning, but general enough to include nearly all programs of practical interest. Like the recently introduced \textit{quasi-Borel predomains}~\cite{vakar2019domain}, our model covers most probabilistic and differentiable programs that can be expressed in today's popular languages: it supports general recursion, higher-order functions, discrete and continuous sampling, and a broad class of primitive functions that includes nearly all the mathematical operations exposed by comprehensive libraries for scientific computing and machine learning (such as \texttt{numpy} and \texttt{scipy}). But crucially, it also \textit{excludes} many pathological functions and kernels that cannot be directly implemented in practice, such as the characteristic function of the Cantor set (whose construction involves an infinite limit that programs cannot compute in finite time). As a result, it is often possible to establish that a desirable property holds of \textit{all} $\omega$PAP maps between two spaces, and to then conclude that it holds of all programs of the appropriate type, without reasoning inductively about their construction or their operational semantics.

The starting point in our search for a semantic model is~\citet{lee2020correctness}'s notion of  \textit{piecewise analyticity under analytic partition} (or PAP). The PAP functions are a particularly well-behaved subset of the functions between Euclidean spaces, and \cite{lee2020correctness} makes a compelling argument that they are ``specific enough'' in the sense we describe above. For example, even though they are not necessarily differentiable, PAP functions admit a generalized notion of derivative, as well as a generalized chain rule that can be used to give a denotational justification of automatic differentiation. These nice properties do not hold for slight generalizations of PAP, e.g. the class of almost-everywhere differentiable functions.

Nevertheless, a reasonable worry is that PAP functions may be \textit{too} narrow a semantic domain: do all programs arising in practice really denote PAP functions, or does the definition assume properties that not all programs enjoy? \citet{lee2020correctness} establish that in a first-order language with PAP primitives and conditionals, all programs denote PAP functions. But real-world programs use many additional features, such as higher-order functions, general recursion, and probabilistic effects. In the presence of these features, do terms of first-order type still denote PAP functions? Consider the program in Fig.~\ref{fig:cantor_func}, for example: it uses recursion to define a partial function that diverges on the $\frac{1}{3}$-Cantor set, and halts on its complement. But we know that the (total) \textit{indicator function} for the $\frac{1}{3}$-Cantor set is not PAP. This may seem like reason to doubt that~\cite{lee2020correctness}'s proposal can be cleanly extended to the general-recursive setting.

\begin{figure}
    \centering
    \includegraphics[width=0.23\textwidth]{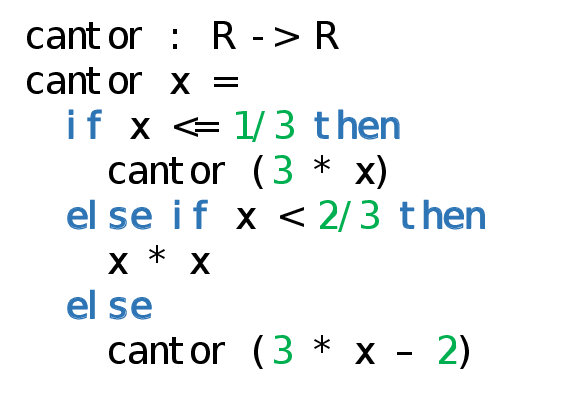}
    \includegraphics[width=0.23\textwidth]{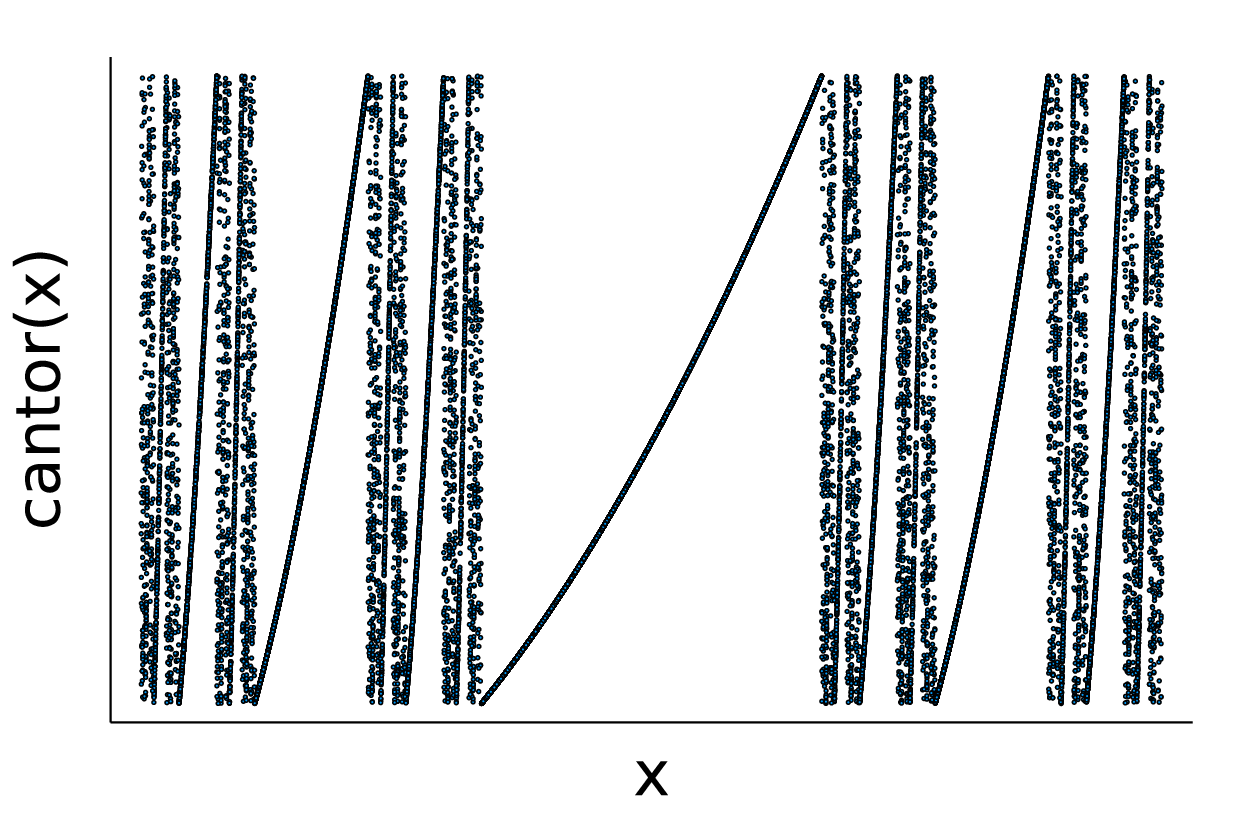}
    \caption{A program that diverges on the $\frac{1}{3}$-Cantor set, but halts elsewhere in $(0, 1)$. The characteristic function of the $\frac{1}{3}$-Cantor set is not PAP, so a naively defined interpretation of ``partial PAP functions into $\RR$'' as PAP functions into $\RR + 1$ would not suffice to interpret a language with general recursion. Our definition of partial $\omega$PAP maps in Section~\ref{sec:semantics} resolves the issue.
    }
    \label{fig:cantor_func}
\end{figure} 

Fortunately, however, it can: in Section~\ref{sec:semantics}, we extend the definition of PAP to cover partial and higher-order functions. The resulting class of functions, the \textit{$\omega$PAP maps}, suffice to interpret all recursive, higher-order programs in an expressive call-by-value PCF, with a type of real numbers and PAP primitives (Theorem~\ref{thm:wpap-sound-adequate}). The $\omega$PAP category also supports a strong monad of measures (Definition~\ref{defn:monad-measure}), enabling an interpretation of continuous sampling and soft conditioning, features common in modern probabilistic languages.

To be clear, existing semantic frameworks can interpret similarly expressive languages (e.g., \cite{vakar2019domain}). Our value proposition is that the $\omega$PAP domain \textit{usefully excludes} certain pathological denotations that other settings allow. After defining our semantics in Section~\ref{sec:semantics}, we devote the remainder of the paper to presenting evidence for this claim: new results (and new proofs of old results) that follow cleanly from denotational reasoning about $\omega$PAP spaces and maps.


\subsection{Differentiable Programming and Automatic Differentiation}

Our first application of our semantics is to the study of \textit{automatic differentiation} (AD), a family of techniques for mechanically computing derivatives of user-defined functions. It has found widespread use in machine learning and related disciplines, where practitioners are often interested in using gradient-based optimization algorithms, such as gradient descent, to fit model parameters to data. Languages with support for automatic differentiation are often called ``differentiable programming languages,'' even if it may be the case that some programs do not denote everywhere-differentiable functions. 

When applied to straight-line programs built from differentiable primitives (without $\iif$~statements, loops, and other control flow), AD has a straightforward justification using the chain rule. But in the presence of more expressive programming constructs, AD is harder to reason about: some programs encode partial or non-differentiable functions, and even when programs \textit{are} differentiable, AD can fail to compute their true derivatives at some inputs.


Our $\omega$PAP semantics gives a clean account of the class of partial functions that recursive, higher-order programs built from ``AD-friendly'' primitives can express. In Section~\ref{sec:differentiable}, we use this characterization to establish guarantees about AD's behavior when applied to such programs:

\begin{itemize}
\item We define a generalized notion of derivative, based on \cite{lee2020correctness}'s \textit{intensional derivatives}, and establish that every first-order $\omega$PAP map is intensionally differentiable.

\item We adapt~\citet{huot2020correctness}'s denotational correctness proof of AD to prove that AD correctly computes intensional derivatives of all programs with first-order type, even if they use recursion and discontinuous primitives (Theorem~\ref{thm:fwd-cor-full-wpap}). This result is stronger than \citet{mazza2021automatic}'s, which follows as a straightforward corollary. But the main benefit over their development is the simplicity of our proof, which mirrors~\citet{huot2020correctness}'s quite closely despite the greatly increased complexity of our language.

\item Using this characterization of what AD computes, we further prove the novel result that intensional gradient descent (which optimizes a \textit{differentiable} function but using gradients computed by AD) converges with probability 1, if initialized with a randomized initial location \textit{and} a randomized learning rate (Theorem~\ref{thm:sgd-ok-full}). This implies that gradient descent can be safely applied to recursive programs with conditionals so long as they denote differentiable functions, even if AD (which can disagree with the true derivative at some inputs) is used to compute gradients.  Interestingly, if either the initial location \textit{or} the learning rate is \textit{not} randomized, a.s.-convergence is not guaranteed, and we present counterexamples of  PyTorch programs that denote differentiable functions but cause PyTorch's gradient descent implementation to diverge (Propositions~\ref{prop:counterexample} \&~\ref{prop:ad-fail-forall-eps}).
\end{itemize}

Taken together, these results give a characterization of the behavior of AD in real systems like PyTorch and Tensorflow, even when they are applied to non-differentiable or recursive programs.


\subsection{Probabilistic Programming}

\textit{Probabilistic programming languages} extend deterministic languages with support for \textit{random sampling} and \textit{soft conditioning}, making them an expressive representation both for probability distributions and for general (unnormalized) measures. Many probabilistic programs are intended to model some aspect of the world, and many questions of interest can be posed as questions about the expectations of certain functions under (normalized versions of) these models. Probabilistic programming languages often come with tools for automatically running \textit{inference algorithms} that estimate such posterior expectations. But the machinery for inference may make assumptions about the programs the user has written\textemdash for example, that the program has a density with respect to a well-behaved reference measure, or that its density is differentiable. Our next application of our semantics is to the problem of verifying that such properties hold for any program in an expressive probabilistic language.

In Section~\ref{sec:probabilistic}, we extend the core language from Section~\ref{sec:semantics} to 
include the $\sample$~and $\score$ commands. This requires changing our monad of effects, which in Section~\ref{sec:semantics} covered only divergence, to also interpret randomness and conditioning. 
We consider two different strong monads on $\omega$PAP: one interprets probabilistic programs as s-finite measures (following~\cite{vakar2019domain}), and the other as weighted sampling procedures (following~\cite{scibior2017denotational,mak2021densities}). In each case, we can prove interesting properties of all probabilistic programs by working just with the class of $\omega$PAP maps they denote:

\begin{itemize}
    \item Interpreting programs as weighted samplers, we prove that almost-surely terminating probabilistic programs have almost-everywhere differentiable weight functions (Theorem~\ref{thm:aediff}). This has been previously shown by~\citet{mak2021densities}, but our proof is substantially simpler, with no need for \cite{mak2021densities}'s stochastic symbolic execution technique. Our theorem is also slightly stronger, covering some programs that do not almost surely halt, and ruling out some pathological a.e.-differentiable weight functions.
    
    \item Interpreting probabilistic programs as $\omega$PAP \textit{measures}, we prove that all programs of type $\RR^n$ denote distributions supported on a countable union of smooth submanifolds of $\RR^n$, and thus admit convenient densities with respect to a particular class of reference measures (Theorem~\ref{cor:main}). 
\end{itemize}

Besides being interesting in their own right, these results have consequences for the practical design and implementation of probabilistic programming systems.

The first result is relevant to the application of gradient-based inference algorithms, which often require (at least) almost-everywhere differentiability to be sound. 

The second result is relevant for the automatic computation of \textit{densities} or \textit{Radon-Nikodym derivatives} of probabilistic programs, key ingredients in higher-level algorithms such as importance sampling and MCMC. To see why, suppose a user has constructed two closed probabilistic programs of type $\RR^n$, $p$ and $q$, and wishes to use $q$ as an importance sampling proposal for $p$. If $\sem{p}$ is absolutely continuous with respect to $\sem{q}$ ($\sem{p} \ll \sem{q}$), such an importance sampler does exist, but implementing it requires computing importance weights: at a point $x \sim \sem{q}$, we must compute $\frac{d\sem{p}}{d\sem{q}}(x)$. If $\sem{p}$ and $\sem{q}$ were both absolutely continuous with respect to the Lebesgue measure on $\mathbb{R}^n$, we could compute probability density functions $\rho_p$ and $\rho_q$ for $\sem{p}$ and $\sem{q}$ respectively, and then compute the importance weight as the ratio of these densities. But not all programs denote measures that \textit{have} densities with respect to the Lebesgue measure.\footnote{Consider, e.g., the program that samples $u \sim \textit{Unif}(0, 1)$, and returns $(u, u) \in \mathbb{R}^2$. Because it is supported only on a 1-dimensional line segment within the plane, it has no density function with respect to the Lebesgue measure in the plane. That is, there is no nonnegative function $\rho$ such that the probability of an event $E \subseteq \RR^2$ is equal to $\iint_{\RR^2} 1_E(x, y) \cdot \rho(x, y) dx dy$.} Our result gives us an alternative to the Lebesgue measure: we can compute the density of $\sem{p}$ (and $\sem{q}$) with respect to the \textit{Hausdorff measures} over the manifolds on which they are supported. (Indeed, computing densities with respect to Hausdorff measures has previously been proposed by~\citet{radul2021base}; our result is that, unlike the Lebesgue base measure, this choice does not \textit{exclude} any possible programs a user might write, because \textit{every} program is absolutely continuous with respect to such a base measure.)

\subsection{Summary}
\begin{itemize}
\item We present the category of $\omega$PAP spaces (Section~\ref{sec:semantics}), a new setting suitable for a denotational treatment of higher-order, recursive differentiable or probabilistic programming languages. It supports sound and adequate semantics of a call-by-value PCF with a type of real numbers and a very expressive class of primitives.

\item To demonstrate that our semantics enables denotational reasoning about a broad class of interesting program properties, we present new results (and new, much-simplified proofs of old results) about differentiable (Section~\ref{sec:differentiable}) and probabilistic (Section~\ref{sec:probabilistic}) programs.
\end{itemize}
\section{\wpap{} Semantics}
\label{sec:semantics}

This section develops the $\wpap{}$ spaces, and shows how to use them to interpret a higher-order, recursive language with conditionals and discontinuous primitives. Our starting point is~\citep{lee2020correctness}'s definition of \textit{piecewise analyticity under analytic partition} (PAP), a property of some functions defined on Euclidean spaces. The $\wpap{}$ construction significantly extends~\citet{lee2020correctness}'s definition to cover \textit{partial} and \textit{higher-order} functions between Euclidean or other spaces. The $\wpap{}$ setting is general enough to interpret almost all primitives encountered in practice, but restricted enough to allow precise denotational reasoning.


\subsection{\wpap{} Spaces}


We first recall a standard definition from analysis:

\begin{definition}[analytic function]
If the Taylor series of a smooth function $f : U \to V$ (for $U \subseteq \RR^n, V \subseteq \RR^m$) converges pointwise to $f$ in a neighborhood around $x$, we say $f$ is \textit{analytic} at $x$. An \textit{analytic function} is a function that is analytic at every point in its domain.
\end{definition}

When a subset of $\RR^n$ can be carved out by finitely many analytic inequalities, we call it an \textit{analytic set}~\citep{lee2020correctness}.

\begin{definition}[analytic set~\citep{lee2020correctness}]
\label{defn:analytic-set}
We call a set $A \subseteq \mathbb{R}^n$ an \textit{analytic set} if there exists a finite collection $\{g_i\}_{i \in I}$ of analytic functions into $\mathbb{R}$, with open domain $U \subseteq \RR^n$, such that 
\[A = \{x \in U \mid \forall i \in I. g_i(x) \leq 0\}.\]
\end{definition}
\noindent(This definition is simpler than in~\cite{lee2020correctness}, but is equivalent.)

\citet{lee2020correctness}'s key definition was the class of \textit{PAP functions}, which are piecewise analytic:

\begin{definition}[PAP function~\citep{lee2020correctness}]For $U \subseteq \RR^n$ and $V \subseteq \RR^m$, we call $f : U \to V$ \textit{piecewise analytic under analytic partition (PAP)} if there is a countable family $\{(A_i, f_i)\}_{i \in I}$ such that:
\begin{enumerate}
    \item the sets $A_i$ are analytic and form a partition of $U$;
    \item each ${f_i : U_i \to \mathbb{R}^m}$ is an analytic function defined on an open domain $U_i\supseteq A_i$;
    \item when $x \in A_i$, $f_i(x) = f(x)$.
\end{enumerate}
\end{definition}
 


In our development, PAP functions will play, at first order, the role  that smooth functions play in the construction of diffeological spaces~\citep{baez2011convenient}, and that measurable functions play in the construction of quasi-Borel spaces~\citep{heunen2017convenient}. The change to PAP is essential: many primitives exposed by languages like TensorFlow and PyTorch fail to be smooth at some inputs, but virtually none fail to be PAP. And while PAP functions are all measurable, they helpfully exclude many pathological measurable functions, making it possible to prove stronger results through denotational reasoning. 

The development that follows generalizes PAP functions to cover higher-order and partial maps, using recently developed categorical techniques~\citep{vakar2019domain,vakar2020denotational,matache2022concrete} (see Appx.~\ref{sec:sound-adequate}).\footnote{Full paper with appendices is available at \url{https://arxiv.org/abs/2302.10636}.} First, we introduce a term for the valid domains of PAP functions.

\begin{definition}[c-analytic set] A set $A \subseteq \mathbb{R}^n$ is \textit{c-analytic} if the inclusion $A \hookrightarrow \mathbb{R}^n$ is PAP.\footnote{Equivalently, $A$ is c-analytic if and only if it is equal to a countable union of (possibly overlapping) analytic subsets of $\RR^n$ (Cor.~\ref{cor:disjointness}).}
\end{definition}

\begin{definition}[\wpap{} space]
An \wpap{} space is a triple $(\conc{X}, \plots{X}, \leq_X)$, where $(\conc{X},\leq)$ is an $\wcpo{}$ and $\plots{X}$ is a family of sets of functions, called plots in $X$. For each c-analytic set $A$, we have
$\plots{X}^A \subseteq \conc{X}^A$. Plots have to satisfy the following closure conditions:
\begin{itemize}
    \item Every map from $\mathbb{R}^0$ to $\conc{X}$ is a plot in $X$.
\item If $\phi \in \conc{X}^A$ is a plot in $X$ and $f : A' \to A$ is a PAP function, then $\phi \circ f$ is a plot in $X$.
\item Suppose the c-analytic sets $A_j \subseteq A$ form a countable partition of the c-analytic set $A$, with inclusions $i_j : A_j \to A$. If $\phi \circ {i_j}$ is a plot in $X$ for every $j$, then $\phi$ is a plot in $X$.
    \item Whenever $(\alpha_i)_i$ is an $\omega$-chain in $\mathcal{P}^A_X$ under the pointwise order, $(\bigvee_i \alpha_i)(x) := \bigvee_i (\alpha_i(x))$ defines a plot in $\plots{X}^A $. 
\end{itemize}
\label{def:wpap}
\end{definition}

\begin{definition}[$\wpap{}$ map]
An \wpap{} map $f : X \to Y$ is a Scott-continuous function between the underlying $\wcpo{}$s $(\conc{X},\leq_X)$ and $(\conc{Y},\leq_Y)$, such that if $\phi\in\plots{X}^A$, $f\circ \phi \in\plots{Y}^A$. 
\end{definition}


We now look at several examples of $\wpap{}$-spaces. Our first example establishes that PAP functions between Euclidean spaces are a special case of $\wpap{}$ maps.

\begin{example}
    Any subset $A \subseteq \RR^n$ can be given the structure of an \wpap{}-space $(A,\plots{A},=)$ where for every c-analytic set $C$, the plots $\plots{A}^C$ are given by
    \[\plots{A}^C:=\{f:C\to A~\mid~f\text{ is a PAP function}\}.\]
\end{example}





Next, we construct products, coproducts, and exponential $\wpap{}$ spaces, which will be useful for interpreting the tuples, sums, and function types of our language.

\begin{example}[Products]
Given $\omega\text{PAP}$ spaces $X$ and $Y$, define $X \times Y$ to be their product as $\wcpo{}$'s, with plots $\langle f, g\rangle \in \mathcal{P}^A_{X \times Y}$ whenever $f \in \mathcal{P}^A_X$ and $g \in \mathcal{P}^A_Y$.

More generally, given  $\omega\text{PAP}$ spaces $X_i$ for $i\in I$, we can define $\prod_{i\in I}X_i$ by their product as $\wcpo{}$s, with plots $f\in \plots{\prod_{i\in I}X_i}^A$ whenever each projection $\pi_i\circ f\in\plots{X_i}^A$.
\end{example}

\begin{example}[Coproducts]\label{example:coproducts}
Let $I$ be a countable index set, and $X_i$ a $\omega\text{PAP}$ space for each $i \in I$. Then $\bigsqcup_i X_i$ is again an $\omega\text{PAP}$ space, with carrier set $\bigsqcup_i |X_i|$, and the partial order inherited from each $X_i$ (elements of different spaces $X_i$ and $X_j$ are not comparable). A function $f : A \to \bigsqcup_i |X_i|$ is a plot if there exists a countable partition $\{A_j\}_j$ of $A$ into c-analytic sets, such that for each $j$, there is a plot $f_j : A_j \to X_{k_j}$, where $k_j \in I$, such that on $A_j$, $f(x) = \mathbf{in}_{k_j} f_j(x)$, i.e., if the preimage of each $X_i$ under $f$ is a c-analytic set for each $i$.
\end{example}

\begin{example}[Exponentials]
Let $X$ and $Y$ be $\omega \text{PAP}$ spaces. Then $X\Rightarrow Y$ is again a $\wpap{}$ space, with carrier set $\wpap{}(X, Y)$, the set of $\wpap{}$-morphisms from $X$ to $Y$. $f\leq_{X\Rightarrow Y} g$ iff for all $x$, $f(x)\leq_Y g(x)$. A function $f : A \to \conc{X\Rightarrow Y}$ is a plot if $\mathbf{uncurry}\, f : A \times X \to Y$ is an $\omega\text{PAP}$ morphism.
\end{example}



\subsection{Core Calculus}

We present our core calculus in Figure~\ref{fig:wpap-language}.
It is a simple call-by-value variant of PCF (e.g. as in \citet{abramsky1998call}), with ground types $\BB$ for Booleans and $\reals$ for real numbers. It is parameterized by a set of primitives $f$ representing PAP functions, including analytic functions such as $+,*,\exp,\tan$ and comparison operators such as $>,=$.
The language also includes constants $\cnst$ for each real $c$ and $\true,\false$ for Booleans.  The typing rules and operational semantics are standard, and given in Appendices~\ref{sub:type-system} and~\ref{sub:op-semantics}.
We will use $\reals^n$ as sugar for the $n$-fold product $\reals\times\ldots\times\reals$. We denote by $\star$ the unique inhabitant of the type $1$.

\begin{figure}[H]
    \fbox{
      \parbox{.47\textwidth}{
\noindent\input{figures/types_wpap}

\noindent\input{figures/terms_wpap}
      }}
\caption{Types and terms of our language}
\label{fig:wpap-language}
\end{figure}

\subsection{Denotational semantics}
\label{sub:denot-semantics}

We give a denotational semantics to our language using $\wpap{}$ spaces and maps. 
A program $\Gamma \vdash \trm:\ty$ is interpreted as an $\wpap{}$ map $\sem{\trm}:\sem{\Gamma}\to \lift \sem{\ty}$, for a suitable partiality monad $\lift$ that we now define.

For an \wpap{} space $Y$, we define the space $\lift Y$ as follows.
\begin{itemize}
    \item The underlying set is $\conc{\lift Y}:=\conc{Y}\sqcup \{\bot\}$.
    \item The order structure is given by $\forall x\in\conc{Y}, \bot\leq_{L Y}x$ and $\forall x,x'\in \conc{Y}, x\leq_{L Y}x'$ iff $x\leq_Y x'$.
    \item Plots are given by 
    \begin{align*}
        \plots{\lift Y}^A &:= \{g:A\to\conc{Y}\sqcup \{\bot\}~|~\exists B\subseteq A~\text{c-analytic};\\
&\qquad\quad g^{-1}(\conc{Y})=B~and~g|_{B}\in\plots{Y}^B\}
    \end{align*}
\end{itemize}

The intuition for $\plots{\lift Y}^A$ is that the \textit{partial} $\wpap{}$ maps into $Y$ are defined on c-analytic sets of inputs, but the region on which they are \textit{not} defined need not satisfy any special property. This is how we account for the example from Fig.~\ref{fig:cantor_func}; although it is undefined on the $\frac{1}{3}$-Cantor set (which is not c-analytic), its domain (the complement of the $\frac{1}{3}$-Cantor set) \textit{is} c-analytic. 

We also have (natural) $\wpap{}$-morphisms, that we can write using lambda-calculus notation as 
\begin{itemize}
\item $\mathbf{return}_X:X\to \lift X$ given by $\lambda x.\mathbf{inl}~x$
\item $\bind_{X, Y}:\lift X \times (X\Rightarrow \lift Y)\to \lift Y$ given by $\lambda (x, g).\mathbf{match} \, x \, \mathbf{with} \, \mathbf{inr}\, \bot \mapsto \mathbf{inr}\, \bot, \, \mathbf{inl}\, x \mapsto g(x)$
\end{itemize}
where $\mathbf{inl},\mathbf{inr}$ are respectively left and right injections (as set functions, they are not $\wpap{}$-morphisms).

$\lift$ extends to \wpap{} maps $f:X\to Y$ by setting $\lift f : \lift X \to \lift Y$ to equal $\mathbf{inl} \circ f$ on all inputs $\mathbf{inl}~x$ and setting $\lift f(\mathbf{inr}~\bot)=\mathbf{inr}~\bot$. 
This definition is inspired by similar constructs present in \citep{vakar2019domain,vakar2020denotational}.
Overall, $(\lift,\mathbf{return},\bind)$ forms a (commutative) monad.

\begin{proposition}
    $(\lift,\mathbf{return},\bind)$ is a (strong, commutative) monad on the category of $\wpap{}$-spaces and morphisms.
\end{proposition}

\noindent We can now interpret types and terms of our language.
The semantics of types and terms is given in Fig.~\ref{fig:wpap-semantics-terms}. 
A context $\Gamma=(\var_1\colon\ty_1,\dots ,\var_n\colon \ty_n)$ is interpreted as the product space
$\sem \Gamma\defeq \prod_{i=1}^n\sem{\ty_i}$. 
Substitution is given by the usual Kleisli composition of effectful programs~\citep{moggi1991notions}.


\begin{figure}[H]
    \fbox{
      \parbox{.48\textwidth}{
      \[
\begin{array}{lcl}
    \sem 1 & \defeq & (\{\star\},\plots{1},=)  \\
    \sem \reals & \defeq & (\RR,\plots{\RR},=)  \\
    \sem{\BB}  & \defeq & (1+1,\plots{1+1},=) \\
    \sem{\tPair{\ty_1}{\ty_2}} & \defeq &\ \sem{\ty_1}\times\sem{\ty_2}\\
    \sem{\ty_1\To\ty_2} & \defeq & \sem{\ty_1}\Rightarrow \lift\sem{\ty_2} 
\end{array}
\]
\[
\setlength\arraycolsep{2pt}
\begin{array}{lcl} 
    \sem{x}(\rho)  & \defeq & \mathbf{return}~\rho(x) \\
    \sem{\star}(\rho)  & \defeq & \mathbf{return}~\star \\
\sem{\underline{c}}(\rho)  & \defeq & \mathbf{return}~c \\
   \sem{f( \trm_1,\dots,\trm_n )}(\rho) & \defeq &
\sem{\trm_1}(\rho)\bind~ (\lambda \var_1.\sem{\trm_2}(\rho)\\
&&\quad\bind \ldots \bind ~(\lambda x_n. \\
&&\quad f(\var_1,\ldots,\var_n)))
\\
\sem{\tPair {\trm_1}{ \trm_2}}(\rho) & \defeq &
\sem{\trm_1}(\rho) \bind~(\lambda x. \sem{\trm_2}(\rho) \bind\\
&&\quad(\lambda y. \mathbf{return}~(x,y)))
\\
\sem{\fun{\var :\ty}{\trm}}(\rho) & \defeq &
\mathbf{return}~\lambda v.\sem {\trm}(\rho[x\mapsto v])
\\ 
\llbracket \llet ~\tPair{\var_1}{\var_2} = \trm_1 \iin~ & \defeq & 
\sem{\trm_1}(\rho) \bind (\lambda v. \sem{\trm_2}(\rho[x_1 \mapsto \pi_1 v,\\
\quad \trm_2 \rrbracket && \quad\quad\quad\quad\quad\quad\quad\quad\quad x_2 \mapsto \pi_2 v])) \\
\sem{\trm_1\,\trm_2}(\rho) & \defeq &
\sem{\trm_1}(\rho) \bind (\lambda f. \sem{\trm_2}(\rho) \bind~f) \\
\llbracket\iif~\trm_1 \then ~\trm_2  & \defeq & \sem{\trm_1}(\rho) \bind \lambda b.\mathbf{if}~b ~\mathbf{then}~\sem{\trm_2}(\rho)\\
 \quad \eelse~ \trm_3 \rrbracket(\rho) && \quad \mathbf{else}~\sem{\trm_3}(\rho) \\
\sem{\mu f.\lambda x.\trm}(\rho) & \defeq & \mathbf{return}~\bigvee_{i \in \mathbb{N}} f_i
\end{array}\]

\quad where $f_0 = \lambda x. \mathbf{inr}~\bot, f_{i+1} = \lambda v. \sem{\trm}(\rho[f \mapsto f_{i}, x \mapsto v])$
      }}
\caption{Denotational semantics in $\wpap{}$. Blue and bold indicate syntactic and mathematical keywords, respectively.}
\vspace{-4mm}
\label{fig:wpap-semantics-terms}
\end{figure}

Our denotational semantics in $\wpap{}$ is sound and adequate w.r.t. our (standard) operational semantics (Appendix~\ref{sub:op-semantics}).

\begin{theorem}
\label{thm:wpap-sound-adequate}
The denotational semantics of our language is sound and adequate.
\end{theorem}

\begin{proof}[Proof sketch]
We show that $\wpap{}$ spaces and morphisms are a category of $\omega$-concrete sheaves on a concrete site, as defined in \citet{matache2022concrete}, and use their development. 
More details are given in Appendix~\ref{sec:sound-adequate}.
\end{proof}


\section{Reasoning about Differentiable Programs}
\label{sec:differentiable}

Our first application of our semantics is to the problem of characterizing the behavior of \textit{automatic differentiation} on the terms of our language. Automatic differentiation is a program transformation intended to convert programs representing real-valued functions into programs representing their derivatives. When all programs denote smooth functions, this informal specification can easily be made precise; for example, Figure~\ref{fig:ad-correct-smooth} illustrates the {correctness property} that~\citet{huot2020correctness} prove for forward-mode AD in a higher-order language with smooth primitives. Their proof is elegant and compositional, relying on logical relations defined over the \textit{denotations} of open terms. 

In our setting, not all programs denote smooth functions, and even when they do, standard AD algorithms can return incorrect results, in ways that depend on exactly how the smooth function in question was expressed as a program. For example,~\citet{mazza2021automatic} define a program $\texttt{SillyId} = \lambda x : \RR. \iif~x = 0~\then~0~\eelse~x$, that denotes the identity function (with derivative $\lambda x. 1$), but for which AD incorrectly computes a derivative of 0 when $x = 0$. Because two extensionally equal programs can have \textit{different} derivatives as computed by AD, it may seem unclear how any proof strategy based on Figure~\ref{fig:ad-correct-smooth}, relating the syntactic operation of AD to the semantic operation of derivative-taking, could apply. Indeed, in showing their correctness result for AD,  \citet{mazza2021automatic} had to develop new operational techniques for reasoning about derivatives of \textit{traces} of recursive, branching programs. 

Our $\omega$PAP semantics lets us take a different route, illustrated in Figure~\ref{fig:ad-correct-pap}. Like~\citet{huot2020correctness}, we work directly with the denotations of terms, which in our setting, are $\wpap{}$ maps. But we do not attempt to assign a unique \textit{derivative} to each $\wpap{}$ function; instead, we define a \textit{relation} on $\wpap{}$ maps that characterizes when one is an \textit{intensional derivative}~\citep{lee2020correctness} of another. Our correctness proof, which follows exactly the structure of~\citet{huot2020correctness}'s, then establishes that AD produces a program denoting \textit{an} (not \textit{the}) intensional derivative of the original program's denotation (Theorem~\ref{thm:fwd-cor-full-wpap}).

All $\omega$PAP maps have intensional derivatives, so there is no need to restrict the correctness result to only the differentiable functions. But when an $\omega$PAP map \textit{is} differentiable in the standard sense, its intensional derivatives agree almost everywhere with its true derivative. Thus, reasoning denotationally, we are able to recover~\citet{mazza2021automatic}'s result: AD is almost everywhere correct. Interestingly, this almost-everywhere correctness result is not sufficient to prove the convergence of randomly initialized gradient descent, even for programs with differentiable denotations (Proposition~\ref{prop:counterexample}).  


\begin{figure}[tb]  
\[
	\xymatrix@C+2mm{
          *+[F]{\txt{Program}}
          \ar[d]_-{\txt{\footnotesize denotational\\\footnotesize semantics}} \ar[rr]^{\txt{\footnotesize automatic\\\footnotesize differentiation}}
		 && *+[F]{\txt{Program}} \ar[d]^-{\txt{\footnotesize denotational\\\footnotesize semantics}}\\
		 *+[F]{\txt{Smooth\\function}}\ar[rr]_{\txt{\footnotesize mathematical\\\footnotesize differentiation}}
		 && *+[F]{\txt{Smooth\\function}}
	}
  \] \caption{Common approach to the correctness of AD, e.g. \citep{huot2020correctness}.
  \label{fig:ad-correct-smooth}}
  \end{figure}

\begin{figure}[tb]
\[
	\xymatrix@C+2mm{
          *+[F]{\txt{Program}}
          \ar[d]_-{\txt{\footnotesize denotational\\\footnotesize semantics}} \ar[rr]^{\txt{\footnotesize automatic\\\footnotesize differentiation}}
		 && *+[F]{\txt{Program}} \ar[d]^-{\txt{\footnotesize denotational\\\footnotesize semantics}}\\
		 *+[F]{\txt{$\omega$PAP\\function}}\ar@{<~>}[rr]_{\txt{\footnotesize intensional\\\footnotesize differentiation}}
		 && *+[F]{\txt{$\omega$PAP\\function}}
	}
  \]
  \caption{Our approach to the correctness of AD.}
  \label{fig:ad-correct-pap}
\end{figure}
  

\subsection{Correctness of Automatic Differentiation}

Terms in our language denote $\wpap{}$ maps $X \to \lift Y$, and so to speak of \textit{correct AD} in our language, we need some notion of derivative that applies to such maps.  We  begin by recalling~\citet{lee2020correctness}'s notion of \textit{intensional derivative}, then lift it to our setting:

\begin{definition}[intensional derivative~\citep{lee2020correctness}]
Let $f : U \to V$ be an $\wpap{}$ map, with c-analytic domain $U \subseteq \RR$ and $V \subseteq \RR$ (seen as $\wpap{}$ spaces). Then an $\wpap{}$ map $g : U \to \RR$ is an \textit{intensional derivative} of $f$ if there exists a countable family $\{(A_i, f_i)\}_{i \in I}$ such that:
\begin{itemize}
\item the sets $A_i$ are analytic and form a partition of $U$;
\item the functions $f_i : U_i \to \RR$ are analytic with open domain $U_i \supseteq A_i$; and
\item for $x \in A_i$, $f(x) = f_i(x)$ and $g(x) = \frac{d}{dx} f_i(x)$. 
\end{itemize}
\end{definition}
\begin{definition}[lifted intensional derivative]\label{def:intensional-deriv-lifted}
Let $f : U \to \lift V$ be an $\wpap{}$ map, with c-analytic domain $U \subseteq \RR$ and $V \subseteq \RR$ (seen as $\wpap{}$ spaces). Then an $\wpap{}$ map $g : U \to \lift \RR$ is a lifted \textit{intensional derivative} of $f$ if it is defined on exactly the inputs that $f$ is, and restricted to this common domain, $g$ is an intensional derivative of $f$.
\end{definition}
\noindent (These definitions can be straightforwardly extended to the multivariate case, where $U \subseteq \RR^n$ and $V \subseteq \RR^m$, and our full correctness theorem does cover such functions. But here we demonstrate the reasoning principles in the simpler univariate setting. Of course, a program of type $\reals \to \reals$ may still be defined in terms of multivariate primitives.)

We now establish that AD, applied to terms with a free parameter $\theta : \reals \vdash t : \reals$, computes these lifted intensional derivatives. It does so in two steps:
\begin{itemize}
\item First, we apply a macro $\ad$ (defined in Fig.~\ref{fig:ad_macro_wpap}) to $t$, yielding a term  $\theta : \reals \times \reals \vdash \ad({t}) : \reals \times \reals$. This new term operates not on a single real, but on a \textit{dual number}, which intuitively stores both a value and its (intensional) derivative with respect to an external parameter. 

\item Then, $\theta : \reals \vdash \llet~\theta=\langle\theta, 1\rangle~\iin~\llet~\langle x,\frac{dx}{d\theta}\rangle= \ad(t) ~\iin~\frac{dx}{d\theta}$ is output as the lifted intensional derivative of $\sem{t}$. 
\end{itemize}

\begin{figure}
\fbox{
  \parbox{.46\textwidth}{
  \[
   \setlength\arraycolsep{1.5pt}
  \begin{array}{lcllcl}
     \ad(\reals) & \defeq & \reals \times \reals & \ad(\tau_1 \times \tau_2) & \defeq & \ad(\ty_1) \times \ad(\ty_2) \\
       \ad(\tau) & \defeq & \tau ~(\tau = 1, \mathbb{B})& \ad(\tau_1 \rightarrow \tau_2) & \defeq & \ad(\ty_1) \to \ad(\ty_2) \\
  \end{array}\]
\[
\begin{array}{lcl}
\ad(\var) & \defeq &  \var\\
 \ad(\star) & \defeq &  \star\\
 \ad(\cnst:\reals) & \defeq &  \tPair{\cnst}{0}\\
 \ad(\cnst:\BB) & \defeq &  \cnst\\
   \ad(f(\trm_1,\ldots,\trm_n)) & \defeq & f_\ad(\ad(\trm_1),\ldots,\ad(\trm_n)) \\
\ad(\fun \var    \trm) & \defeq &  \fun\var{\ad(\trm)}\\ 
\ad(\trm\, s ) & \defeq &  
\ad(\trm)\,\ad(s)\\
\ad(\tPair{\trm_1}{\trm_2}) & \defeq &  \tPair{\ad(\trm_1)}{\ad(\trm_2)} \\
\ad({\tMatch{\trm}{\var_1,\var_2}{s}}) & \defeq & 
\tMatch{\ad(\trm)\\&&}{\var_1,\var_2}{\ad(s)} \\
\ad(\mu f.\trm) & \defeq &  \mu f.\ad(\trm) \\
\ad(\iifthenelse{\trm_1}{\trm_2}{\trm_3}) & \defeq & \iifthenelse{\ad(\trm_1)}{\ad(\trm_2)\\&&}{\ad(\trm_3)}
\end{array}
\]
}}
\caption{AD macro for our higher-order recursive language}
\label{fig:ad_macro_wpap}
\end{figure}

The key hurdle in showing the correctness of this AD procedure is establishing that the macro $\ad$ is correctly propagating intensional dual numbers. To show this compositionally, and in a way that exploits our denotational framework, we first define for each type $\tau$ a \textit{relation} $V_\tau$ between pairs of $\omega$PAP \textit{plots} in $\sem{\tau}$ and in $\sem{\ad({\tau})}$. It captures more precisely the properties that the dual numbers flowing through our program should have.

\begin{definition}[correct dual-number intensional derivative at $A \subseteq \RR$]
Let $\tau$ be a type and let $f : A \to \sem{\tau}$, for some c-analytic subset $A \subseteq \RR$. Then we say $f' : A \to \sem{\ad(\tau)}$ is a \textit{correct dual-number intensional derivative} of $f$ if $(f, f') \in V_\tau(A)$, where $V_\tau(A)$ is defined inductively:
\begin{itemize}
\item $V_{\reals}(A) = \{(f, f') \mid f'(x) = (f(x), g(x))$ for some intensional derivative $g$ of $f\}$.
\item $V_\mathbb{\tau}(A) = \{(f, f') \mid f = f'\}$ for $\tau \in \{1, \mathbb{B}\}$.
\item $V_{\tau_1 \times \tau_2}(A) = \{(f, f') \mid (\pi_i \circ f, \pi_i \circ f') \in V_{\tau_i}(A)$ for $i = 1, 2\}$.
\item $V_{\tau_1 \to \tau_2}(A) = \{(f, f') \mid \forall (g, g') \in V_{\tau_1}(A),$ $(\lambda x. f(x)(g(x)), \lambda x. f'(x)(g'(x))) \in V_{\tau_2}(A)\}$. 
\end{itemize}
\end{definition}
Readers may recognize $V$ as defining a (semantic) logical relation; its definition at product and function types is completely standard, so only at ground types did we make any real choices. For partial functions, we need to extend the definition:
\begin{definition}[correct lifted dual-number intensional derivative at $A \subseteq \RR$]
Let $\tau$ be a type and let $f : A \to \lift\sem{\tau}$, for some c-analytic subset $A \subseteq \RR$. Then we say $f' : A \to \lift \sem{\ad(\tau)}$ is a \textit{correct lifted dual-number intensional derivative} of $f$ if $f$ and $f'$ are defined (i.e., not $\bot$) on the same domain and, restricted to this common domain, $f'$ is a correct dual-number intensional derivative of $f$.
\end{definition}
We are now in a position to state more precisely what it would mean for the $\ad$ macro to correctly propagate dual numbers. In particular, the macro transforms terms $\Gamma \vdash t : \tau$ into their \textit{correct dual-number translations}  $\ad(\Gamma) \vdash \ad(t) : \ad(\tau)$, the specification for which is given below.

\begin{definition}[correct dual-number translation]\label{def:correct-translation}
Let $\tau_1$ and $\tau_2$ be types, and let $f : \sem{\tau_1} \to \lift \sem{\tau_2}$. Then $f_{\ad} : \sem{\ad(\tau_1)} \to \lift\sem{\ad(\tau_2)}$ is a correct dual-number translation of $f$ if, for all c-analytic $A$, and all pairs $(g, g') \in V_{\tau_1}(A)$, $f_D \circ g'$ is a correct lifted dual-number intensional derivative of $f \circ g$.
\end{definition}

We assume that each primitive $f : \tau_1 \to \tau_2$ in our language is equipped with a \textit{built-in} correct dual-number translation $f_{\ad}$, which the definition of $\ad$ in Figure~\ref{fig:ad_macro_wpap} uses. Because the primitives  are translated correctly, so are whole programs:

\begin{lemma}[Correctness of $\ad$ macro (limited)]\label{lem:fundamental}
For any term $\Gamma \vdash t : \tau$, $\sem{\ad(t)} : \sem{\ad(\Gamma)} \to \sem{\ad(\tau)}$ is a correct dual-number translation of $\sem{t}$. 
\end{lemma}
\begin{proof}[Proof Sketch]
This is the Fundamental Lemma of our logical relations proof, and can be obtained using general machinery~\citep{katsumata2013relating}, after verifying several properties of our particular setting:
\begin{itemize}
\item $V_{\reals}(A)$ is closed under \textit{piecewise gluing}: if $\{A_i\}_{i \in I}$ is a countable partition of $A$, and $(f|_{A_i}, f'|_{A_i}) \in V_{\reals}(A_i)$ for each $i \in I$, then $(f, f') \in V_{\reals}(A)$. This follows from our definition of intensional derivative.
\item $V_{\reals}(A)$ is closed under \textit{restriction}: if $(f, f') \in V_{\reals}(A)$, and $A' \subseteq A$ is c-analytic, then $(f|_{A'}, f'|_{A'}) \in V_{\reals}(A')$. This also follows easily from our definition.
\item The lifted dual-number intensional derivatives are closed under \textit{least upper bounds}: if $\{(f_n, f'_n)\}_{n \in \mathbb{N}}$ is a sequence with each $f'_n : A \to \lift\sem{\RR \times \RR}$ a correct lifted dual-number intensional derivative of $f : A \to \lift \RR$, and with $f_n \leq f_{n+1}$ and $f'_n \leq f'_{n+1}$ for each $n \in \mathbb{N}$, then $\bigvee_{n \in \mathbb{N}} f'_n$ is a correct lifted dual-number intensional derivative of $\bigvee_{n \in \mathbb{N}} f_n$. This is more involved, but its proof mirrors the one that establishes $\lift \RR$ as an $\wpap{}$ space in the first place: we explicitly construct representations of the lubs in question as piecewise functions, satisfying Definition~\ref{def:intensional-deriv-lifted}.
\end{itemize}
\end{proof}

The overall correctness theorem follows easily:
\begin{theorem}[Correctness of AD Algorithm (limited)]
  \label{thm:fwd-cor-basic-wpap}
  For any term $\theta : \reals \vdash \trm : \reals$, the $\wpap{}$ map
  $\lambda \theta. \sem {\llet~\langle x,\frac{dx}{d\theta}\rangle 
 = \ad(t)~\iin~\frac{dx}{d\theta}}((\theta, 1))$ is a lifted intensional derivative of $\sem{\trm}$.
\end{theorem}

\begin{proof}[Proof]
We apply Lemma~\ref{lem:fundamental}, setting $(g, g')$ from Definition~\ref{def:correct-translation} to $(\lambda \theta. \theta, \lambda \theta. (\theta, 1))$. This implies that $(\sem{t}, \lambda \theta. \sem{\ad(t)}(\theta, 1))$ are defined on the same domain, and on that domain $A$, their restrictions are in $V_{\reals}(A)$. This implies that on that domain, $\lambda \theta. \pi_2(\sem{\ad(t)}((\theta, 1)))$ is an intensional derivative of $\sem{t}$, from which the result follows by unfolding the definition of $\sem{\cdot}$ on $\llet$ expressions.
\end{proof}

This result, which we achieved via denotational reasoning, immediately implies the result of~\citet{mazza2021automatic}:

\begin{corollary}
If $\theta : \reals \vdash t : \reals$ denotes a total differentiable function, then the set of input points on which the AD algorithm fails has measure zero.\footnote{\citet{mazza2021automatic} prove a slightly stronger result, and we can also show something slightly stronger; see Prop.~\ref{propn:failure-set-quasivariety}.}
\end{corollary}
\begin{proof}
Because $\sem{t}$ is total, it is PAP and AD computes an intensional derivative of $\sem{t}$. Intensional derivatives are almost everywhere equal to true derivatives~\citep{lee2020correctness}.
\end{proof}


Using more complex logical predicates, we can prove correctness for multivariate functions (see Appx.~\ref{sec:correctness-ad}). In the multivariate setting, forward-mode AD typically computes \textit{Jacobian-vector products}; for us, it computes lifted \textit{intensional} Jacobian-vector products.
\begin{definition}[lifted intensional Jacobian]
Let $f : U \to \lift V$ be an $\wpap{}$ map, with $U \subseteq \RR^n$ c-analytic and $V \subseteq \RR^m$. An $\wpap{}$ map $g : U \to \lift \RR^{m \times n}$ is an \textit{intensional Jacobian} of $f$ if it is defined (i.e., not $\bot$) exactly when $f$ is, there exists a countable analytic partition $\{A_i\}_{i \in I}$ of its domain, and there exist analytic functions $f_i : U_i \to \RR^m$ with open domains $U_i \supseteq A_i$ such that when $x \in A_i$, $f(x) = f_i(x)$ and $g(x) = Jf_i(x)$, where $J$ is the Jacobian operator.
\end{definition}

\begin{theorem}[Correctness of AD (full)]
  \label{thm:fwd-cor-full-wpap}
  For any term $x_1: \reals, \dots, x_n: \reals\vdash\trm: \reals^m$, and any vector $\mathbf{v} \in \RR^n$, the $\wpap{}$ map 
  $\lambda \mathbf{\theta}. \textbf{let } \mathbf{x} = \sem {\ad(\trm)}((\theta_1, v_1), \dots, (\theta_n, v_n)) \textbf{ in } (\pi_2 x_1, \dots, \pi_2 x_m)$ is equal to $g(\mathbf{\theta}) \cdot \mathbf{v}$ for some intensional Jacobian $g$ of $\sem{t}$.
\end{theorem}

\subsection{Optimization of Differentiable Programs with AD-computed Derivatives}

Our characterization of AD's behavior, proven in the previous section, can be used to establish the correctness of algorithms that use the results of AD. In optimization, for example, a very common technique for minimizing a differentiable function $f:\RR^d\to\RR$ is gradient descent. Starting at an initial point $x^0 \in \RR^d$, for a set step size $\epsilon > 0$, we iterate the following update:

\begin{equation}
\label{eqn:gd-step}
    x^{t+1} := x^t-\epsilon \grad f(x^t).
\end{equation}

If the step size $\epsilon$ is too large, we cannot hope to converge on an answer. To ensure that a small-enough step size can be found, we need to know that $f$'s gradient changes slowly:

\begin{definition}
     A differentiable function $f:\RR^d\to\RR$ is $L$-smooth if for all $x,y\in\RR^d$
     \[\norm{\grad f(x)-\grad f(y)}\leq L\norm{x-y}. \]
\end{definition}

When $f$ is $L$-smooth, it is well known that $\epsilon$ can be chosen small enough for gradient descent to converge (albeit not necessarily to a global minimum):

\begin{theorem}[e.g. \citep{polyak1987introduction}]
\label{thm:gd-conv}
   Let $f$ be $L$-smooth, bounded below, and $0<\epsilon<\frac{2}{L}$. Then, for any initial $x^0\in\RR^d$, following the method from Equation~\ref{eqn:gd-step}, the gradient of $f$ tends to 0:
   \[\lim_{t\to\infty}\grad f(x^t) =0\]
   and the function $f$ monotonically decreases on values of the sequence $\{x^t\}_t$, that is for all $t\in\NN$, $f(x^{t+1})\leq f(x^t)$.
\end{theorem}

However, this theorem assumes that the true gradient of $f$ is used to make each update. When gradients are obtained by AD, this may not be the case, even if $f$ is differentiable: we proved only that AD computes \textit{intensional derivatives}, which may disagree with true derivatives on a measure-zero set of inputs. 

Intuitively, it may seem as though we can avoid the measure-zero set of inputs on which AD may be incorrect by \textit{randomly} selecting an initial point $x^0$. Indeed, similar arguments have been made informally in the literature~\citep{mazza2021automatic}. But this turns out not to be the case:

\begin{proposition}\label{prop:counterexample}
    Let $\epsilon > 0$. There exists a program $P$ such that $\sem{P}:\RR\to\RR$ 
    satisfies the condition of Theorem~\ref{thm:gd-conv}, and yet, for all $x^0 \in \RR$, the gradient descent method diverges to  $-\infty$, when the gradients are computed with AD.
\end{proposition}

\begin{proof}[Proof sketch]
   Let $P$ be the closed program defined by the following code, where $lr$ is our chosen learning rate $\epsilon$.
   
    \begin{minted}[fontsize=\small]{python}  
    # Recursive helper
    def g(x, n):
        if x > 0:
            return ((x - n)*(x - n))/(lr * 2)
        if x == 0:
            return x / lr + n*n / (2*lr)
        else:
            return g(x+1, n+1)
    
    def P(x): 
      return g(x,0)        
    \end{minted}
    One can show that $\sem{P}=x \mapsto \frac{x^2}{2 \epsilon}$ and satisfies the conditions of Theorem~\ref{thm:gd-conv}, yet gradient descent will diverge to $-\infty$ for all initial points $x^0$. 
    The full explanation is given in Appendix~\ref{sub:failure-ad}.
\end{proof}

The proof applies to an idealized setting where gradient descent is run with real numbers, rather than floating-point numbers, but in our experiments, PyTorch did diverge on this program.
Similar counterexamples can be constructed for some variants of the gradient descent algorithm, e.g. when the learning rate is changing according to a schedule.

Likewise, it is easy to derail the algorithm when the learning rate $\epsilon$ is random but the initial point $x^0$ is fixed.

\begin{proposition}
\label{prop:ad-fail-forall-eps}
    Let $x^0$ be a fixed initial point. Then
    there exists a program $P$ such that $\sem{P}:\RR\to\RR$ 
    satisfies the conditions of Theorem~\ref{thm:gd-conv}, and $\grad \sem{P}(x^0) \neq 0$, and yet,  for all $\epsilon>0$, the gradient descent method yields $x^t = x^0$ for all $t$, when the gradients are computed with AD.
\end{proposition}

\begin{proof}
 Simply choose $P(x) = \iifthenelse{x = x^0}{(x^0)^2}{x^2}$, for which AD will give the intentional derivative 0 at $x^0$.
\end{proof}

The counterexamples constructed for the proofs above are tied to a specific learning rate or to a specific initialization. It is therefore reasonable to think that by randomizing \textit{both} quantities, one might be able to almost surely avoid such counter-examples. Thankfully, this \textit{is} the case.

\begin{theorem}[Convergence of GD with intensional derivatives]
\label{thm:sgd-ok-full}
     Let $f$ be PAP, bounded below, $L$-smooth.
     Let $(\epsilon,x_0)$ be drawn randomly from a continuous probability distribution supported on $(0,\frac{2}{L})\times \RR^d$. Then when gradient descent computes gradients using AD, the true gradient tends to 0 with probability 1:
   \[\lim_{t\to\infty}\grad f(x^t) =0.\]
   Furthermore, with probability 1, the function $f(x)$ monotonically decreases: $f(x^{t+1})\leq f(x^t)$.
\end{theorem}

The proof is given in Appendix~\ref{sub:a.s-conv-gd}.
\section{Reasoning about Probabilistic Programs}
\label{sec:probabilistic}

PAP functions were originally developed in~\citet{lee2020correctness} to reason about automatic differentiation, so perhaps it is not surprising that their generalization to $\wpap{}$ was useful, in the previous section, for characterizing AD's behavior. In this section, we show that the $\omega$PAP semantics \textit{also} makes it nice to reason about \textit{probabilistic} programs: by excluding pathological deterministic primitives, we also prevent many pathologies from arising in the probabilistic setting, enabling clean denotational arguments for several nice properties.

To reason about probabilistic programs, we must extend our core calculus (Figure~\ref{fig:wpap-language}) with the standard effectful operations exposed by probabilistic programming languages (PPLs):
$\sample$ samples a real from the uniform distribution on $(0,1)$, and  $\score~r$ reweights program traces (to implement soft constraints). 
Their typing rules are given in Figure~\ref{fig:type_system_proba}.

\begin{figure}
\centering



\begin{tabular}{c}
     $\Gamma\vdash e:\reals$  \\ \hline
     $\Gamma\vdash \score~e: 1$ 
\end{tabular}
\qquad
\begin{tabular}{c}
      \\ \hline
     $\Gamma \vdash \sample : \reals$ 
\end{tabular}

\caption{Type system of the probabilistic part of the language}
\label{fig:type_system_proba}
\end{figure}

The language needs no new types, but the semantics change: instead of interpreting our language with the monad $\lift$, which models divergence, we must use a new monad that also tracks probabilistic computation. 
In Section~\ref{sub:trace-semantics}, we introduce a monad of \textit{weighted samplers}~\citep{scibior2017denotational}, which can informally be seen as deterministic functions from a source of randomness to output values. We use this monad to reason about the deterministic \textit{weight functions} that some probabilistic programming systems use during inference, presenting a new (much simplified) proof of~\citet{mak2021densities}'s result that almost-surely-terminating probabilistic programs have almost-everywhere differentiable weight functions. 
In Section~\ref{sub:base-measure}, we introduce a commutative monad of measures, which can be seen as a smaller version of~\citet{vakar2019domain}'s monad of measures, excluding those measures which can't be defined using $\wpap{}$ deterministic primitives. We use this smaller monad to prove that all definable measures in our language are supported on countable unions of smooth manifolds. This implies that they have densities with respect to a particular well-behaved class of base measures, and we close by briefly discussing the implications for PPL designers.



\subsection{Almost-Everywhere Differentiability of Weight Functions}
\label{sub:trace-semantics}

First, we define a monad $\mathbf{S}$ of \textit{weighted samplers}, that interprets programs as partial $\wpap{}$ maps from a space $\Omega$ of random seeds to 
a space of weighted values. Intuitively, the measure represented (or ``targeted'') by a sampler of weighted pairs $(x, w)$ is the one that assigns measure $\mathbb{E}[w \cdot 1_A(x)]$ to each measurable set $A$. But this semantics does \textit{not} validate many important program equivalences, like commutativity, intuitively because it exposes ``implementation details'' of weighted samplers, that can differ even for samplers targeting the same measure. This is intentional: our first goal is to prove a property of these ``implementation details,'' one that inference algorithms automated by PPLs might rely on.

First, we fix our source of randomness $\Omega$ to be the lists of reals, $\sqcup_{i \in \mathbb{N}} \RR^i$ (which is an $\wpap{}$ space, see Example~\ref{example:coproducts}). We then define our monad:

\begin{definition}[monad of weighted samplers]
Let $\mathbf{S}$ be the monad $(\mathbf{S}, \mathbf{return}^\mathbf{S}, \bind^\mathbf{S})$, defined as follows:
\begin{itemize}
\item $\mathbf{S}X := \Omega \Rightarrow \lift ((X + 1) \times [0, \infty) \times \Omega)$
\item $\mathbf{return}^\mathbf{S}_X := \lambda x. \lambda r. \mathbf{inl} (\mathbf{inl}~x, 1, r)$
\item $\bind^\mathbf{S}_{X,Y} := \lambda m. \lambda k. \lambda r. m(r) \bind^{\mathbf{L}} (\lambda (x_{?}, w, r'). \mathbf{match}~x_{?}$ $\mathbf{with}~\{\mathbf{inr}~\_ \to \mathbf{inl}~(x_{?}, 0 , r') \mid \mathbf{inl}~x \to k(x)(r') \bind^{\lift} \lambda (y_?, v, r''). \mathbf{inl}~(y_{?}, w\cdot v, r'')\}$.
\end{itemize}
\end{definition}

\begin{figure}[H]
    \fbox{
      \parbox{.48\textwidth}{
\[
\setlength\arraycolsep{3pt}
\begin{array}{lcl}
    \sem{\sample}(\rho)  & \defeq & \lambda r. \mathbf{match}~r~\mathbf{with} \{\\
    &&\quad() \to \mathbf{inl}~(\mathbf{inr}~\star, 0, r) \\
    &&\quad \mid (r_1, r_{2:n}) \to \mathbf{inl}~(\mathbf{inl}~r_1, 1, r_{2:n})\}\\
    \sem{\score~t}(\rho)  & \defeq & \sem{t}(\rho) \bind (\lambda w. \lambda r. \mathbf{inl} (\mathbf{inl}~\star, 0 \vee w, r)) \\
\sem{\mu f.\lambda x. \trm}(\rho) & \defeq & \mathbf{return}~\bigvee_{i \in \mathbb{N}} f_i
\end{array}\]

where $f_0 = \lambda x. \lambda r. \mathbf{inr}~\bot, f_{i+1} = \lambda v. \sem{\trm}(\rho[f \mapsto f_{i}, x \mapsto v])$
      }}
\caption{Updated semantics for new monad $\mathbf{S}$ of effects}
\label{fig:semantics-prob-trace}
\end{figure}

This monad interprets a probabilistic program as a partial function mapping lists of real-valued random samples, called traces, to: (1) the output values they possibly induce in $X$ (or $\mathbf{inr}~\star$ if the input trace does not contain enough samples to complete execution of the program), (2) a \textit{weight} in $[0, \infty)$, and (3) a remainder of the trace, containing any samples not yet used. The $\mathbf{return}$ of the monad consumes no randomness, and successfully returns its argument with a weight of $1$. The $\bind$ of the monad runs its continuation on the random numbers remaining in the trace after its first argument has executed, and multiplies the weights from both parts of the computation. Figure~\ref{fig:semantics-prob-trace} gives the semantics of $\sample$ and $\score$, as well as updated semantics for $\mu f. t$ (the only change is that the bottom element, from which least fixed points are calculated, is now the bottom element of $\mathbf{S}X$ rather than $\lift X$). The $\sample$ command fails to produce a value when given an empty trace, but otherwise consumes the first value in the trace and returns it. The $\score$ command consumes no randomness, and uses its argument (if non-negative) as the weight.

One way to use a weighted sampler is to run it, continually providing longer and longer lists of uniform random numbers until a value from  $X$ (and an associated weight) is generated. But many probabilistic programming systems also expose more sophisticated inference algorithms. \textit{Gradient-based} methods, such as Hamiltonian Monte Carlo or Langevin ascent, attempt to find executions with high weights by performing hill-climbing optimization or MCMC over the space of traces, guided by a program's \textit{weight function}.

\begin{definition}[weight function]
For any $p \in |\mathbf{S}~X|$, define its \textit{weight function} $w_p : \sqcup_{i \in \mathbb{N}} [0, 1]^i \to [0, \infty)$, as follows:

\begin{itemize}
\item When $p(r) = \mathbf{inr}~\bot$, $w_p(r) = 0$.
\item When $p(r) = \mathbf{inl}~(x, v, r')$ for $r'$ non-empty, $w_p(r) = 0$.
\item When $p(r) = \mathbf{inl}~(x, v, ())$, $w_p(r) = v$.
\end{itemize}
\end{definition}

The weight function can be viewed as a density with respect to a particular base measure on traces; scaling the base measure by this density yields an unnormalized measure, and the normalized version of this measure, which is a probability distribution placing high mass on traces that lead to high weights, is called the \textit{posterior} over traces. Inference algorithms like HMC are designed to generate samples approximately distributed according to this posterior.

Using our definition of the weight function, we can recover~\citet{mak2021densities}'s result about the almost-everywhere differentiability of weight functions, a convenient property when reasoning about gradient-based inference:

\begin{theorem}\label{thm:aediff}
Probabilistic programs that almost surely halt have almost-everywhere differentiable weight functions.
\end{theorem}
\begin{proof}
Let $\vdash t : X$ be a program. Almost-sure termination implies that for almost all inputs $r \in \sqcup_{i \in \mathbb{N}} [0,1]^i$, $\sem{t}(r) \neq \mathbf{inr}~\bot$. Restricted to the intersection of $\sqcup_{i \in \mathbb{N}} [0,1]^i$ with $\sem{t}$'s domain, its weight function $w_{\sem{t}}$ is $\wpap{}$,\footnote{Under our semantics, $\sem{t}$ is an $\wpap{}$ morphism from $1$ (the empty context) to $\mathbf{S}\sem{\tau}$ for any closed term $t$. Elements of $\mathbf{S}\sem{\tau}$ are themselves $\wpap{}$ morphisms from $\Omega$ to $\mathbf{L}((\sem{\tau}+1) \times \RR_{\geq 0} \times \Omega)$ — so $t$ denotes such a map. Now consider post-composing this map with the $\wpap{}$ map sending $\bot$ to $\bot$, and $(x, v, r)$ to $v$ when $r$ is empty and $0$ otherwise. The result is an $\wpap{}$ map from $\Omega$ to $\mathbf{L} \RR_{\geq 0}$. We next precompose this map with the inclusion $\sqcup_i [0,1]^i \to \Omega$, to obtain an $\wpap{}$ map from $\bigsqcup_i [0,1]^i$ to $\mathbf{L}\RR_{\geq 0}$, equal to the weight function for $\sem{t}$ everywhere on its domain.} 
and thus almost-everywhere differentiable. Almost all of almost all of  $\sqcup_{i \in \mathbb{N}} [0,1]^i$ is almost all of $\sqcup_{i \in \mathbb{N}} [0,1]^i$, concluding the proof.
\end{proof}

We note that the proof only uses the fact that for almost all $r \in \sqcup_{i \in \mathbb{N}} [0,1]^i$, $\sem{t}(r) \neq \mathbf{inr}~\bot$\textemdash which is strictly weaker than almost-sure termination. Almost-sure termination can fail even though a program halts on any finite trace of random numbers. Some probabilistic context-free grammars, for example, do not always surely halt, but as they execute, they consume an unbounded amount of randomness, so that when provided a finite trace, they will eventually halt with an error. Our proof shows that such programs, although they do not almost surely halt, do have almost-everywhere differentiable weight functions. As such, our result is slightly stronger than that of~\citet{mak2021densities}, and followed immediately from the interpretation of probabilistic programs in the $\wpap{}$ category. 

\subsection{Existence of Convenient Base Measures for Monte Carlo}
\label{sub:base-measure}



We now wish to prove a result not about the intensional properties of probabilistic programs (like their weight functions), but the extensional properties of the measures that programs denote. To do so, we construct a final monad in which to interpret our language: a commutative monad of measures, similar to that of~\citet{vakar2019domain}.

\begin{definition}[$\wpap{}$ space of measure weights]
Define $\mathbb{W}$ to be the $\wpap{}$ space with:
\begin{itemize}
\item $|\mathbb{W}| = \RR_{\geq 0} \cup \{\infty\}$, the extended non-negative reals;
\item $w \leq_\mathbb{W} v$ when $v = \infty$ or when $w < v < \infty$.
\item $\phi \in \plots{\mathbb{W}}^{A}$ whenever $\phi : A \to \mathbb{W}$ is measurable.
\end{itemize}
\end{definition}
Note that although many maps \textit{into} $\mathbb{W}$ exist, very few maps exist \textit{out of} $\mathbb{W}$, because its $\wpap{}$ diffeology is so permissive: any measurable function is a plot.

Following~\citet{scibior2017denotational}, measures can be seen as equivalence classes of samplers that behave the same way \textit{as integrators}. Given a weighted sampler $p \in |\mathbf{S} X|$, we can define the \textit{integrator} it represents:

\begin{definition}[integrator associated to a sampler]
\label{defn:integrator_sampler}
We define $\wpap{}$ maps $\mathbf{Int}_X : \mathbf{S} X \to (X \Rightarrow \mathbb{W}) \Rightarrow \mathbb{W}$ sending samplers to their associated \textit{integrators}: $\mathbf{Int}_X~p = \lambda f. \int_{\text{Dom}(\alpha_p)} \Lambda_{\Omega}(dr) f(\alpha_p(r))$, where
\begin{itemize}
\item $\Lambda_\Omega$ is a Lebesgue base measure over $|\Omega|$, assigning to each measurable subset $A \subseteq \Omega$ the value $\sum_{i \in \mathbb{N}} \Lambda^i(\{\mathbf{x} \mid (i, \mathbf{x}) \in A\})$ (where $\Lambda^i$ is the Lebesgue measure on $\RR^i$).
\item $\alpha_p : \Omega \to \lift X$ is an $\wpap{}$ map that on input $r$, returns $\mathbf{inr}~\bot$ if any of the following hold: $r$ is empty; $r_1 < 0$; $r_{2:|r|} \not\in [0,1]^{|r|-1}$; $p(r_{2:|r|})$ returns $\mathbf{inr} \bot$; or $p(r_{2:|r|})$ returns $\mathbf{inl}~(x,v,r')$ and either $r'$ is non-empty, $x = \mathbf{inr}~\star$, or $v < r_1$. Otherwise, $p(r_{2:|r|})$ must return $\mathbf{inl}~(\mathbf{inl}~x, \_, \_)$, and $\alpha_p$ also returns $\mathbf{inl}~x$.
\item $\text{Dom}(\alpha_p)$ is the domain of the partial map $\alpha_p$.
\end{itemize}
\end{definition}

Then, following~\citet{vakar2019domain}, we take the measures to be those integrators that lie in the image of $\mathbf{Int}$:
\begin{definition}[monad of measures]
\label{defn:monad-measure}
The monad $\mathbf{M}$ of measures is defined by:
\begin{itemize}
\item $\mathbf{M}X$ is the \textit{image} of $\mathbf{S}X$ under $\mathbf{Int}_X$. It is a subobject of $(X \Rightarrow W) \Rightarrow W$ and inherits its order.
\item $\mathbf{return}^\mathbf{M}_X := \lambda x. \lambda f. f(x)$.
\item $\bind^\mathbf{M}_{X,Y} := \lambda m. \lambda k. \lambda f. m(\lambda x. k(x)(f))$.
\end{itemize}
\end{definition}

\begin{figure}[H]
    \fbox{
      \parbox{.48\textwidth}{
\[
\setlength\arraycolsep{3pt}
\begin{array}{lcl}
    \sem{\sample}(\rho)  & \defeq & \lambda f. \int_{[0,1]} f(x) dx\\
    \sem{\score~t}(\rho)  & \defeq & \lambda f. \sem{t}(\rho)(\lambda w. (0 \vee w) \cdot f(\star))\\
\sem{\mu f. \lambda x. \trm}(\rho) & \defeq & \mathbf{return}~\bigvee_{i \in \mathbb{N}} f_i
\end{array}\]

\quad where $f_0 = \lambda x. \lambda f. 0, f_{i+1} = \lambda v. \sem{\trm}(\rho[f \mapsto f_{i}, x \mapsto v])$
      }}
\caption{Updated semantics for new monad $\mathbf{M}$ of effects}
\label{fig:semantics-prob-meas}
\end{figure}

This monad is a submonad of the continuation monad~\citep{vakar2019domain}, and so the $\mathbf{return}$ and $\bind$ are inherited. The $\sample$ command is interpreted as the uniform measure on $[0, 1]$, and the $\score~t$ command as a measure on a one-point space, with a certain total mass given by $t$. Our semantics under $\mathbf{M}$ (Fig.~\ref{fig:semantics-prob-meas}) is related to the one from the previous section (Fig.~\ref{fig:semantics-prob-trace}) in that for any closed term $\vdash t : X$, $\mathbf{Int}_X(\sem{t}_{\mathbf{S}}) = \sem{t}_\mathbf{M}$. Using this, we can establish this lemma about the definable measures:

\begin{lemma}
\label{thm:trace-semantics-is-pap}
 Let $\vdash \trm : \reals^n$. 
 Then, there exists a partial \wpap{} function $f_{\trm} : \Omega \rightharpoonup \RR^n$ such that $\sem{\trm} = {f_{\trm}}_* \Lambda_\Omega$.
\end{lemma}

That is, any probabilistic program returning reals must arise as a well-behaved transformation of ``input randomness'' represented by $\Omega$. This is not particularly surprising, of course, because a probabilistic program is \textit{defined} by specifying an ($\wpap{}$) transformation of input randomness, but it highlights the way we will use our $\wpap{}$ semantics to establish general properties of the measures denoted by probabilistic programs. 
The proof is given in Appendix~\ref{sub:monad-of-measures}.



In particular, we will characterize the definable measures as \textit{absolutely continuous} with respect to a particular class of base measures. We first recall some relevant definitions from differential geometry, starting with the \textit{smooth manifolds}, a well-studied generalization of Euclidean spaces; common examples include lines, spheres, and tori.

\begin{definition}
     A smooth manifold $M$ is a second-countable Hausdorff topological space together with a smooth atlas: an open cover $\mathcal{U}$ together with homeomorphisms $(\phi_U:U\to \RR^n)_{U\in\mathcal{U}}$ called charts such that $\phi_V\circ \phi_U^{-1}$ is smooth on its domain of definition for all $U,V\in\mathcal{U}$. A function $f:M\to N$ between manifolds is smooth if  $\phi_V\circ f\circ \phi_U^{-1}$ is smooth for all charts $\phi_U,\phi_V$ of $M$ and $N$. 
When there's a global bound $K$ on the local dimensions $n$ of a smooth manifold, it's well known by a theorem of Whitney that the manifold can be embedded as a submanifold of an Euclidean space, that is its topology is given by restricting the standard one from the surrounding Euclidean space.
\end{definition}

There is a natural measure on smooth manifolds that generalizes the Lebesgue measure on Euclidean spaces, and matches our intuitive notions of area and volume on curved spaces.

\begin{definition}
 The $k$-Hausdorff measure on a metric space $(X,d)$ is defined as
 \vspace{-1mm}
{\footnotesize \[H^k(S) := \lim_{\delta \to 0}\inf_{(U_i)_{i \in \mathbb{N}}} \Big\{\sum_{i=1}^\infty\diam(U_i)^k\mid S\subseteq \bigcup_{i=1}^{\infty} U_i, \diam(U_i)<\delta\Big\}\]}

\noindent where $\diam(U):=\sup~\{d(x,y)~\mid~x,y\in U\}$ and $\diam(\emptyset):=0$. 
The Hausdorff dimension is then defined as $\dim_{\text{Haus}}(S):=\inf \{k~\mid~ H^k(S)=0\}$. 
\end{definition}
The Hausdorff measure computes the sizes of sets $S$ in a dimension-dependent way: 
\begin{itemize}
\item The $0$-Hausdorff measure counts the points in $S$. 
\item The $1$-Hausdorff measure sums the lengths of curves within a set. It will be infinite on surfaces, and 0 on sets of isolated points.
\item The $2$-dimensional Hausdorff measure sums the areas of surfaces, and will be infinite on volumes and 0 on curves. 
\item More generally, the $k$-Hausdorff measure will quantify $k$-dimensional volumes.
\end{itemize}

For smooth manifolds, the Hausdorff dimension matches the intuitive notion of dimension, e.g. a sphere in $\RR^3$ has Hausdorff dimension $2$ , a path on such a sphere will have Hausdorff dimension $1$, and an open of $\RR^n$ has dimension $n$.
We now define our class of \textit{s-Hausdorff} measures, which measure the sizes of sets by measuring the $d$-dimensional volumes of their intersections with $d$-dimensional manifolds, for every $d$:

\begin{definition}[s-Hausdorff measure on $\RR^n$]
\label{def:shausdorff}
    An s-Hausdorff measure on $\RR^n$ is a measure $\mu$ that decomposes as
    \[\mu(A)= \sum_{d=0}^n H^d(A\cap M^d)\]
    where each $M^d$ is a countable union $\bigcup_i M^d_i$ of $d$-dimensional smooth manifolds $M^d_i$.
\end{definition}



Our main result in this section is a construction assigning an s-Hausdorff measure $B(\trm)$ to every closed probabilistic program $\vdash \trm : \reals^n$, such that $\sem{\trm}$ has a density w.r.t. $B(\trm)$. 

\begin{theorem}
\label{cor:main}
Any closed probabilistic program $\vdash \trm: \reals^n$ admits
an s-Hausdorff measure $B(\trm)$ on $\RR^n$ such that
\begin{itemize}
    \item $\sem{\trm}$ has a density (possibly infinite) $\rho_{\trm}$ w.r.t $B(\trm)$;
    \item if $\mu$ is another s-Hausdorff measure w.r.t. which $\sem{\trm}$ has a density, then that density is $\sem{\trm}$-a.e. equal to $\rho_{\trm{}}$.
    \item if $t_2$ has the same type and $\sem{\trm_1} \ll \sem{\trm_2}$ (where $\ll$ denotes absolute continuity) then $B(\trm_1) \ll B(\trm_2)$; and
    \item there is a set $A \subseteq \RR^n$ such that $\mathbf{1}_{A}$ is a density of $B(\trm_1)$ with respect to $B(\trm_2)$. 
\end{itemize}
\end{theorem}

In particular, this means that a sensible design decision for a PPL is to always compute densities of probabilistic programs $p$ with respect to an s-Hausdorff base measure $B(p)$. In order to compute Radon-Nikodym derivatives between multiple programs, the theorem implies we can separately compute their densities and then take the ratio, so long as we also include a ``support check'' (of the sort described by~\citep{radul2021base}).

Finally, the following result shows that the class of s-Hausdorff measures seems appropriate to serve as base measures, as for each s-Hausdorff measure, there is a closed probabilistic program whose posterior distribution has a strictly positive density w.r.t that measure. This means that we cannot `carve out' any mass from the s-Hausdorff measure, and that the underlying supporting manifolds faithfully represent the possible supports of posterior distributions of probabilistic programs.

\begin{theorem}
\label{thm:definability-s-hausdorff}
    Assume that the language has all partial PAP functions as primitives. Then, for every $n\in\NN$ and every s-Hausdorff measure $\mu$, there exists a program $\vdash \trm:\reals^n$ such that $\sem{\trm}$ and $\mu$ are mutually absolutely continuous.
\end{theorem}

The proofs are given at the end of Appendix~\ref{sub:pushforward-pap}.


\section{Related Work and Discussion}
\label{sec:discussion}

\noindent\textbf{Summary.} We introduced the category of $\wpap{}$ spaces and used it as a denotational model for expressive higher-order recursive effectful languages. 
We first looked at a deterministic language with conditionals and partial PAP functions, which we argued covers almost all differentiable programs that can be implemented in practice. We showed that AD computes correct intentional derivatives in such an expressive setting, extending \citet{lee2020correctness}'s result, and recovering the fact that AD computes derivatives which are correct almost-everywhere. Next, we showed that gradient descent on programs implementing differentiable functions can be soundly used with intentional derivatives, as long as both the learning rate and the initialization and randomized. 
We then looked at applications in probabilistic programming. After introducing a strong monad capturing traces of probabilistic programs, we gave a denotational proof that the trace density function of every probabilistic program is almost everywhere differentiable. 
Finally, we defined a commutative monad of measures on $\wpap{}$, and proved that all programs denote measures with densities with respect to some \textit{s-Hausdorff measures}. As such, we argued that the s-Hausdorff measures form a set of convenient base measures for density computations in Monte Carlo inference, and showed that every closed probabilistic program has a density w.r.t. some s-Hausdorff measure.

Together, these results demonstrate the value of denotational reasoning with a ``just-specific-enough'' model like $\wpap{}$. Results previously established in the literature by careful operational reasoning, such as Theorem~\ref{thm:aediff}, follow almost immediately after setting up the definitions the right way. And we have also shown new theorems, such as Theorem~\ref{cor:main}, by applying methods from analysis to characterize the restricted class of denotations in our semantic domain.
\\\vspace{-3mm}

\noindent\textbf{Semantics of higher-order differentiable programming.} We continue a recent line of work on giving denotational semantics for higher-order differentiable programming in functional languages. Our semantics and logical relations proof builds on insights proposed in the smooth setting in \citet{huot2020correctness}, and takes inspiration from \citet{vakar2020denotational} for the treatment of recursion.
In this work, we considered a simple forward-mode AD translation, but a whole body of work focused their attention on provably correct and efficient reverse-mode AD algorithms on higher-order languages \citep{brunel2019backpropagation,vakar2022chad,huot2022rad,mazza2021automatic,krawiec2022provably}. 
We believe our correctness result could be adapted to a reverse-mode AD algorithm, perhaps following the neat ideas developed in \citet{krawiec2022provably}, \citet{radul2022you}, or  \citet{smeding2023efficient}, but we leave this for future work. 
There are also more synthetic approaches to studying differentiation in a more purely categorical tradition, sometimes called synthetic differential geometry \citep{cockett2014differential}. Some of these approaches are particularly appealing from a theoretical point of view \citep{cockett2019reverse,blute2010convenient}, but their precise relations with AD and its implementations remains to be further studied.
\\\vspace{-3mm}

\noindent\textbf{Semantics of differentiable programming with non-smooth functions.} Our work directly builds on and extends \citet{lee2020correctness}'s setting to a higher-order recursive language. As is shown in \citet{lee2023smoothness}, when moving beyond the differentiable setting, 
there are many seemingly reasonable classes of functions that behave pathologically in subtle ways, usually preventing compositional reasoning. 
\wpap{} spaces combine the advantage of restricting to PAP functions at first-order with expressive power provided by abstract categorical constructions, which conservatively extend the first-order setting. \citet{bolte2020mathematical} investigated a similar problem as \citet{lee2020correctness} on a more restricted first-order language, but proved a convergence of gradient descent result in their setting.
It would be interesting to see if some of the ideas developed in \citet{chizat2018global} could be adapted to prove some convergence of stochastic gradient descent when AD is used on programs denoting PAP functions, thus going beyond the setting of neural networks. 
\\\vspace{-3mm}

\noindent\textbf{Semantics of probabilistic programming.} Our commutative monad of measures takes clear inspiration from \citet{staton2017commutative} and \citet{vakar2019domain}, adding an extra step in the recent search of semantic models for higher-order probabilistic programming \citep{heunen2017convenient,ehrhard2017measurable}.
In particular, our work refines the model of \citet{vakar2019domain} by restricting to PAP functions, instead of merely measurable ones, but keeping its essential good properties for interpreting probabilistic programs. By doing so, our work is closer in spirit to \citet{freer2012computable}'s study of computable properties for higher-order probabilistic languages. 
Our monad for tracking traces of probabilistic programs is similar to what is presented in \citet{lew2019trace}, giving a denotational version of similar work that focus on operational semantics, such as \citet{mak2021densities}. Our result about densities and s-Hausdorff measures has some of its foundations based on the careful study of s-finite measures and kernels from \citet{vakar2018s}.
It would be interesting to see if the work of \citet{lew2023adev} could be extended to a non-differentiable setting with PAP primitives, proving that an AD algorithm on a probabilistic language yields unbiased gradient estimates, perhaps using the estimator derived in \citet{lee2018reparameterization} for non-smooth functions. 
\\\vspace{-3mm}

\noindent\textbf{Disintegration and base measure.}
Our result on base measures for probabilistic programs is related to  the literature on symbolic disintegration \citep{shan2017exact,cho2019disintegration,narayanan2020symbolic}, which has also had to wrestle with the problem of finding base measures with respect to which expressive programs have densities. Our approach builds on the idea presented by~\citet{radul2021base}: we extend the Hausdorff base measures to s-Hausdorff base measures, and prove that they suffice for ensuring all programs have densities.
We leave open the question of characterizing exactly the class of densities on s-Hausdorff measures that arise from probabilistic programs, and the investigation of the closure of these measures under least upper bounds. 

\ifCLASSOPTIONcompsoc
  \section*{Acknowledgments}
  We have benefited from discussing this work with many friends and colleagues, Wonyeol Lee, Faustyna Krawiec, Michele Pagani, Sean Moss, and the Oxford group. 
We are also grateful to anonymous referees for their very helpful feedback.
This material is based on work supported by the NSF Graduate
Research Fellowship under Grant No. 1745302.
Our work is also supported by a Royal Society University Research Fellowship, the ERC BLAST grant, the Air Force Office of Scientific Research (Award No. FA9550–21–1–0038), and the DARPA Machine Common Sense and SAIL-ON projects.
\else
  \section*{Acknowledgment}
  We have benefited from discussing this work with many friends and colleagues, Wonyeol Lee, Faustyna Krawiec, Michele Pagani, Sean Moss, and the Oxford group. 
We are also grateful to anonymous referees for their very helpful feedback.
This material is based on work supported by the NSF Graduate
Research Fellowship under Grant No. 1745302.
Our work is also supported by a Royal Society University Research Fellowship, the ERC BLAST grant, the Air Force Office of Scientific Research (Award No. FA9550–21–1–0038), and the DARPA Machine Common Sense and SAIL-ON projects.
\fi

\ifCLASSOPTIONcaptionsoff
  \newpage
\fi



%

\bibliographystyle{IEEEtranN} 
\bibliography{refs.bib}




%








\newpage

\appendices
\section{Language}
\subsection{Type System}
\label{sub:type-system}

The type-system for the deterministic language (Section~\ref{sec:differentiable}) is given in Figure~\ref{fig:type_system}.
 
\begin{figure}[H]
\centering

 \begin{tabular}{c}
      \\\hline
     $\Gamma \vdash x:\tau$ 
\end{tabular}\,$(x:\tau\in\Gamma)$
\quad
\begin{tabular}{c}
$\Gamma,x:\tau_1 \vdash \trm:\tau_2$ \\\hline
$\Gamma \vdash \lambda x:\tau_1.\trm : \tau_1\to\tau_2$
\end{tabular}
\vspace{.2cm}

\begin{tabular}{c}
      \\ \hline
      $\Gamma \vdash c:\tau$
\end{tabular}$(c\in\mathcal{C}_{\tau})$
\quad
\begin{tabular}{c}
     $\Gamma \vdash \trm_1:\tau_1\to\tau_2$ 
     \quad $\Gamma \vdash \trm_2:\tau_1$  \\ \hline
     \quad $\Gamma \vdash \trm_1\trm_2 :\tau_2$
\end{tabular}
\vspace{.2cm}

\begin{tabular}{c}
      $\Gamma \vdash \trm_1:\tau_1$ \quad  $\Gamma \vdash \trm_1:\tau_2$  \\ \hline
      $\Gamma \vdash \tPair{\trm_1}{\trm_2}:\tau_1\times \tau_2$
\end{tabular}
\quad
\begin{tabular}{c}
      \\ \hline
     $\Gamma \vdash \star:1$ 
\end{tabular}
\vspace{.2cm}

\begin{tabular}{c}
     $\Gamma\vdash \trm: \tau_1\times\tau_2$ \quad $\Gamma,\var_1:\tau_1,\var_2:\tau_2 \vdash s: \tau$  \\ \hline
     $\Gamma\vdash {\tMatch{\trm}{\var_1,\var_2}{s}}:\tau $
\end{tabular}
\vspace{.2cm}

\begin{tabular}{c}
     $\Gamma \vdash \trm_1: \BB$ 
     \quad  $\Gamma \vdash \trm_2:\tau$
     \quad  $\Gamma \vdash \trm_3:\tau$ \\ \hline
     $\Gamma \vdash \iifthenelse{\trm_1}{\trm_2}{\trm_3}: \tau$
\end{tabular}
\vspace{.2cm}

\begin{tabular}{c}
   $\Gamma,f:\tau_1\to\tau_2 \vdash \trm:\tau_1\to\tau_2 $ \\ \hline
   $\Gamma \vdash \mu f:\tau_1\to\tau_2. \trm:\tau_1\to\tau_2$
\end{tabular}

\caption{Type system of the deterministic part of the language. $\mathcal{C}_\tau$ denotes the set of primitives of type $\tau$.}
\label{fig:type_system}
\end{figure}



\subsection{Operational semantics}
\label{sub:op-semantics}

Values are given by the following grammar
\[v::= \var\mid c\mid f \mid \lambda x.\trm \mid \mu f.\trm \mid \tPair{v_1}{v_2} \mid \star \]

A big step-semantics is given in Figure~\ref{fig:op_semantics}.

\begin{figure}[H]
\centering
\begin{tabular}{c}
      \\ \hline
     $\cnst \bigstep \cnst$
\end{tabular}
\quad
\begin{tabular}{c}
     $\trm_2 \bigstep \lambda \var.\trm$ \quad $\trm_1\bigstep v_1$ \quad $\trm{}[v_1/\var{}]\bigstep v_2$ \\ \hline
     $\trm_2\trm_1 \bigstep v_2$
\end{tabular}
\vspace{.2cm}

\begin{tabular}{c}
  $\forall i, \trm_i\bigstep v_i$ \quad $\trm \bigstep f$ \\ \hline
  $\trm(\trm_1,\ldots,\trm_n)\bigstep \underline{f(v_1,\ldots,v_n)}$
\end{tabular}
\vspace{.2cm}

\begin{tabular}{c}
     $\trm_1\bigstep \true$ \quad $\trm_2\bigstep v$ \\ \hline
     $\iifthenelse{\trm_1}{\trm_2}{\trm_3}\bigstep v$
\end{tabular}
\vspace{.2cm}

\begin{tabular}{c}
     $\trm_1\bigstep \false$ \quad $\trm_3\bigstep v$ \\ \hline
     $\iifthenelse{\trm_1}{\trm_2}{\trm_3}\bigstep v$
\end{tabular}
\vspace{.2cm}

\begin{tabular}{c}
   $\trm\bigstep \tPair{v_1}{v_2}$ \quad $\trm_2[v_1/\var_1,v_2/\var_2]\bigstep v$ \\ \hline
      $\pMatch{\trm}{\var_1}{\var_2}{\trm_2}\bigstep v$
\end{tabular}
\vspace{.2cm}


\begin{tabular}{c}
    $\trm{}[\mu f.\trm/ f] \bigstep v$ 
    \\ \hline
   $\mu f.\trm\bigstep v$
\end{tabular}
\caption{Big step operational semantics of the deterministic part of the language}
\label{fig:op_semantics}
\end{figure}
\section{Sound and adequate denotation}
\label{sec:sound-adequate}

\subsubsection{Concrete sites and sheaves}

Our proof of soundness and adequacy uses reusable machinery developed by~\citet{matache2022concrete}. Their framework requires establishing various properties; for convenience, in this section, we reproduce the relevant definitions, and recall the theorem that we will use.

\begin{definition}[Reproduced from \citep{matache2022concrete}]
A \textit{concrete category} is a category $\catC$ with a terminal object $ \term$ such that the functor $\catC(\term,-):\catC\to\Set$ is faithful. 
\end{definition}

\begin{example}
    Let $\OpenCont$ be the category whose objects are open subsets $U$ of $\RR^n$ (for all $n\in\NN$), and whose morphisms $U\to V$ are continuous functions. $\RR^0$ is a terminal object, and $\OpenCont(\RR^0,-)$ simply sends an open set to its underlying set, and sends a continuous function to its underlying function. Therefore, $\OpenCont(\RR^0,-)$ is faithful and $\OpenCont$ is concrete.
\end{example}
 
\begin{definition}[Reproduced from \citep{matache2022concrete}]
A \textit{concrete site} $(\catC,\topo)$ is a small concrete category $\catC$ with an initial object $ \ini$, together with a \textit{coverage}  $ \topo$, which specifies for each object $c$ a set $ \topo(c)$ of families  of maps with codomain $c$ . We call such a family $\{f_i: c_i\to c\}_{i\in I}\in\topo(c)$ a \textit{covering family} and say it covers $c$. The coverage must satisfy the following axioms.
\begin{enumerate}
    \item For every map $h:d\to c $ in $\catC$, if $\{f_i: c_i\to c\}_{i\in I}$ covers $c$, then there is a covering family $\{g_j: d_j\to d\}_{j\in I'}$ of $d$ such that every $h\circ g_j$ factors through some $f_i$.
    \item If  $\{f_i: c_i\to c\}_{i\in I}$ covers $c$, then $\bigcup_{i\in I}Im(\conc{f_i})=\conc{c}$.
    \item the initial object $\ini$ is covered by the empty set.
    \item The identity is always covering 
    \item If  $\{f_i: c_i\to c\}_{i\in I}\in\topo(c)$ and  $\{g_{ij}: c_{ij}\to c_i\}_{ j\in J_i}\in\topo(c_i)$ for each $i$, then $\{f_i\circ g_{ij}:c_{ij}\to c\}_{i\in I,j\in J_i}\in \topo(c)$.
\end{enumerate} 
\end{definition}

\begin{example}[Cartesian spaces and smooth maps]
    $\Cartsp$ is a concrete site, where $\{f_i:U_i\to V\}_i$ is a covering family if $\bigcup_i f_i(U_i)=V$, i.e. the images of the $f_i$ cover the codomain $V$ in the usual sense. 
\end{example}

\begin{example}
     Similarly to $\Cartsp$, $\OpenCont$ is a concrete site, where $\{f_i:U_i\to V\}_i$ is a covering family if $\bigcup_i f_i(U_i)=V$, i.e. the images of the $f_i$ cover the codomain $V$ in the usual sense. 
     Its initial object is the empty set.
\end{example}

\begin{example}[Standard Borel spaces]
    The category $\Sbs$ has objects the Borel subsets of $\RR$ and morphisms the measurable functions between these objects. It is a concrete site where the coverage contains the countable sets of inclusions functions $\{U_i\to U\}_i$ such that $U=\bigcup_i U_i$ and the $U_i$ are disjoint.
\end{example}

\begin{definition}[Reproduced from \citep{matache2022concrete}]
A \textit{concrete sheaf} $X$ on a concrete site $(\catC,\topo)$ is a set $\conc{X}$, together with, for each object $c\in\catC$, a set $\plots{X}^c$ of functions $\conc{c}\to\conc{X}$, such that
\begin{itemize}
    \item each $\plots{X}^c$  contains all the constant functions;
    \item for any map $h:d\to c\in\catC$, and any $g\in\plots{X}^c $, the composite function $g\circ \conc{h}:\conc{d}\to\conc{X}$ is in $\plots{X}^d$;
    \item for each function $g:\conc{c}\to\conc{X}$ and each covering family $\{f_i: c_i\to c\}_{i\in I}$, if each $g\circ \conc{f_i}\in\plots{X}^{c_i}$, then $g:\conc{c}\to\conc{X}\in\plots{X}^c$.
\end{itemize}

A morphism $\alpha:X\to Y$ between concrete sheaves is a function $\alpha:\conc{X}\to \conc{Y}$ that preserves the structure, namely if $g\in\plots{X}^c$, then $\alpha\circ g\in\plots{Y}^c$.
\end{definition}

\begin{example}
    Let $\Cont$ be the category defined as follows. Its objects are pairs $(X,\plots{X})$ where for each open $U\subseteq \RR^n$, $\plots{X}^U\subseteq X^U$ is closed under the following axioms:
    \begin{itemize}
        \item Constant functions are in $\plots{X}^U$
        \item If $f\in \plots{X}^U$ and $g:V\to U$ is continuous, then $f \circ g\in \plots{X}^V$
        \item If $\{U_i\}_{i\in\NN}$ is an open cover of an open $U$, and for all $i$, $f_i\in \plots{X}^{U_i}$, such that for all $i,j\in \NN^2$, $f_i$ and $f_j$ agree on $U_i\cap U_j$, then $f:U\to X$ defined by $x\in U_i\mapsto f_i(x)$, is an element of $\plots{X}^U$.
    \end{itemize}
    Its morphisms $X\to Y$ are functions $\alpha:X\to Y$ such that $\alpha \circ f\in \plots{Y}^U$ whenever $f\in\plots{X}^U$.
Then, $\Cont$ is a category of concrete sheaves on the concrete site $\OpenCont$.
\end{example}

\begin{example}
    Diffeological spaces and Quasi-Borel spaces are examples of categories of concrete sheaves and morphisms between concrete sheaves. Other examples are given in \citet{baez2011convenient}.
\end{example}

\subsubsection{\ensuremath{\omega}-Concrete sheaves}

To interpret recursion, the concrete-sheaf structure will not be sufficient in general. Instead of evoking what is to my knowledge the yet-to-be-formalised theory of enriched sheaves over the category of $\wcpo{}$'s, we continue with the setting from \citet{matache2022concrete}.

\begin{definition}[Reproduced from \citep{matache2022concrete}]
An $\omega$-concrete sheaf on a site $\site$ is a concrete sheaf $X$ together with an ordering $\leq_X$ on $\conc{X}$ that equips $X$ with the structure of\wcpo, such that each $\plots{X}^c$ is closed under pointwise suprema of chains with respect to the pointwise ordering.
A morphism $\alpha:X\to Y$ of $\omega$-concrete sheaves is a continuous function between \wcpo's, $\alpha:\conc{X}\to \conc{Y}$, that is also a morphism of concrete sheaves.
\end{definition}

$\omega$-concrete sheaves form a category $\wconcsh$, which is a Cartesian-closed category with binary coproducts \citep{matache2022concrete}.

\begin{example}
    The category of $\omega$-diffeological spaces \citep{vakar2020denotational} and of $\omega$-Quasi-Borel spaces \citep{vakar2019domain} are examples of categories of $\omega$-concrete sheaves.
\end{example}

\subsubsection{Partiality: admissible monos and a partiality monad}
 
To model recursion, we first need to define a (strong) partiality monad $\lift$ on $\wconcsh$.
There are many choices for the lifting of the plots, so we parametrize the definition of partiality monad by a class $\monos$ of monomorphisms from the site $\site$, which we call admissible monos.

Recall that, in any category, monos with the same codomain are preordered. 
If $m_1: d_1\mono c, m_2:d_2\to c$ then $m_1\leq m_2$ iff there exists $f:d_1\to d_2$ such that $m_2 \circ f= m_1$
We write $Sub(c)$ for the poset quotient of the set of monos with codomain $c$. For any class $\monos$ of monos, we write $\submono(c)$ for the full subposet of $Sub(c)$ whose elements have representatives in $\monos$.

\begin{definition}[Reproduced from \citep{matache2022concrete}]
A class $\monos$ of admissible monos for a concrete site $\site$ consists of, for each object $c\in \catC$, a set of monos $\monos(c)$ with codomain $c$ satisfying the following conditions
\begin{enumerate}
    \item For all $c\in\catC, 0\to c\in\monos(c)$
    \item $\monos$ contains all isomorphisms
    \item $\monos$ is closed under composition
    \item All pullbacks of $\monos$-maps exist and are again in $\monos$ 
    (This makes $\submono$ a functor $\catC^{op}\to\Set$).
    \item For each $c$, the function $\submono(c)\to\Set(\conc{c},\{0, 1\})$ is componentwise injective and order-reflecting, and the image of $\submono$ is closed under suprema of $\omega$-chains
    \item Given an increasing chain in $\submono(c)$, $(c_n\mono c)_{n\in\NN}$, denote its least upper bound by $\cinf \mono c$. Then the closure under precomposition (with any morphism) of the set $\{c_n\mono \cinf\}_{n\in\NN}$ contains a covering family of $\cinf$.
\end{enumerate}
\end{definition}

Now, given a class $\monos$ of admissible monos, we can define a  partiality monad $L_\monos$ on the category of $\omega$-concrete sheaves:

\begin{definition}[Reproduced from \citep{matache2022concrete}]
We define the strong partiality monad $\lift_\monos$ associated to the class of admissible monos $\monos$ as
\begin{itemize}
    \item $\conc{\lift_\monos X}=\conc{X}\sqcup \{\bot\}$
    \item $\forall x\in\conc{X}, \bot\leq_{\lift_\monos X}x$
    \item $\forall x,x'\in \conc{X}, x\leq_{\lift_\monos X}x'$ iff $x\leq_X x'$
    \item $\plots{\lift_\monos X}^c=\{g:\conc{c}\to\conc{X}\sqcup\{\bot\}~|~\exists c'\mono c\in\monos(c)~s.t.~g^{-1}(\conc{X})=Im(\conc{c'})~and~g|_{Im(\conc{c'})}\in\plots{X}^{c'}\}$
\end{itemize}
The strong monad structure is exactly the same as the maybe monad on $\Set$.
\end{definition}

\begin{proposition}[\citep{matache2022concrete}]
\label{prop:lift-is-monad}
    $\lift_\monos:\wconcsh\to\wconcsh$ is a strong monad.
\end{proposition}

\subsubsection{Soundness and adequacy}

\citet{matache2022concrete} defined a language PCF$_v$, a call-by-value variant of PCF. It has product, sum types, and a recursion operator similar to ours, and is presented as a fine-grained call-by-value language.  
This language is essentially equivalent to our language, except that we have extra base types $\reals$ and $\BB$, and more basic primitives involving these types. Our conditional is definable using their sum types. They give their language a standard call-by-value operational semantics, which translates directly to an operational semantics for our language.
In summary, even though their theorem is stated for PCF$_v$, it directly applies to our language, without needing further changes or extra proofs, as long as we also interpret our base types and primitives in the model, as in standard \citep{moggi1991notions}.

\begin{theorem}[\citep{matache2022concrete}]
\label{thm:wconshsound}
A concrete site with a class of admissible monos,  $\sitewmono$, is a sound and adequate model of PCF$_v$  in  $\wconcsh$.
\end{theorem}

Having recalled the general setting of $\omega$-concrete sheaves, we now show how \wpap{} is an instance of this general construction. 
To do so, we will show define the category $\cpap$ and show it is a concrete-site (which is the one for \wpap{}) in Section~\ref{sub:cpap}. We will define a set of admissible monos for $\cpap$ in Section~\ref{sub:admin-monos-cpap}. In order to achieve this, we will need a key technical proposition (Proposition~\ref{prop:c-analytic-disjoint}) and a few other technical lemmas, proved in Section~\ref{sub:c-analytic-disjoint}.

\subsubsection{C-analytic sets} 
\label{sub:c-analytic-disjoint}

The definition of c-analytic sets does not require that the analytic sets that make up a c-analytic set be disjoint. Still, it turns out that it is always possible to partition a c-analytic set into countably many pairwise disjoint analytic sets. 

\begin{lemma}
    Definition~\ref{defn:analytic-set} of analytic sets is equivalent to the following definition from \citet{lee2018reparameterization}.
    
    $A \subseteq \mathbb{R}^n$ is an analytic set if there exists finite collections $\{g^-_i :U^-_i\to \RR\}_{i \in I}$ and $\{g^+_j:U^+_j\}_{j \in J}$ of analytic functions on open domains, such that 
    \begin{align*}
        A &= \big\{x \in \big( \cap_{i\in I} U_i^- \big)\cap \big(\cap_{j\in J} U_j^+\big) \mid \forall i \in I. g^-_i(x) \leq 0,\\
        &\qquad\qquad\qquad\forall j \in J. g^+_j(x)>0\big\}
    \end{align*}
\end{lemma}

\begin{proof}
\begin{itemize}
    \item  $U:= \big( \cap_{i\in I} U_i^- \big)\cap \big(\cap_{j\in J} U_j^+\big) $ is a finite intersection of open sets and is open, and we can restrict each analytic function to $U$.
    \item Given $g: U\to \RR$, the inequality $g(x)>0$ can be changed to the inequality $f(x)\leq 0$ where $f: V\to \RR$ defined by $f(x)=- g(x)$, and $V:= U \cap \{x\in U~\mid~g(x)\neq 0\}$ is open.
\end{itemize}  
\end{proof}

We use \citet{lee2018reparameterization}'s definition in the rest of the appendix.

\begin{proposition}
\label{prop:c-analytic-disjoint}
For any c-analytic set $A$, there exists a countable collection of pairwise disjoint analytic sets $\{B_i\}_{i \in \mathbb{N}}$ such that $A = \cup_{i \in \mathbb{N}} B_i$. 
\end{proposition}

This is an important property, and most of this section is devoted to proving the proposition.

Let $\{A_i\}_{i \in \mathbb{N}}$ be a countable set of analytic subsets of $\mathbb{R}^n$. 

Each analytic set $A_i$ is associated with an open domain $U_i$ and a finite number $J_i$ of analytic functions $\{f_{ij}\}_{j \in [1, \dots, J_i]}$, so that $A_i = \{x \in U_i \mid \forall j. f_{ij} \leq 0\}$. Because any open set is a countable union of open balls in $\mathbb{R}^n$, we can assume without loss of generality that the $U_i$ are open balls.

\begin{definition}
A \textit{simple analytic set} is a set $A = \{x \in \mathcal{B}_\epsilon(x_0) \mid \forall i \leq I. g_i^+(x) > 0 \wedge \forall j \leq J. g_i^-(x) \leq 0\}$, for natural numbers $I$ and $J$ and some $0 < \epsilon \leq \infty$ and $x_0 \in \mathbb{R}^n$.
\end{definition}

\begin{lemma}
The complement of any simple analytic set is a finite union of pairwise disjoint simple analytic sets.
\end{lemma}

\begin{proof}
A point $x$ can fail to be in a simple analytic set $A$ for $I + J + 1$ mutually exclusive reasons, each of which applies to a simple analytic set of points:
\begin{enumerate}
    \item If $\epsilon < \infty$, it can fail to lie in $\mathcal{B}_\epsilon(x_0)$. The set of points outside $\mathcal{B}_\epsilon(x_0)$ is simple analytic: let $h_1^-(x) = \epsilon - ||x - x_0||_2$ and consider $\{x \in \mathbb{R}^n \mid h_1^-(x) \leq 0\}$.
    \item (For each $1 \leq i \leq I$.) It can lie within $\mathcal{B}_\epsilon(x_0)$, and satisfy $g_n^+(x) > 0$ for all $n < i$, but fail to satisfy $g_i^+(x) > 0$. Let $h_n^+(x) = g_n^+(x)$ for $n < i$, and $h_1^-(x) = g_i^+(x)$. Then the set of all such points is $\{x \in \mathcal{B}_\epsilon(x_0) \mid \forall n < i. h_n^+(x) > 0 \wedge h_1^-(x) \leq 0\}$.
    \item (For each $1 \leq j \leq J$.) It can lie within $\mathcal{B}_\epsilon(x_0)$, and satisfy $g_i^+(x) > 0$ for all $i \leq I$, and satisfy $g^-_n(x) \leq 0$ for all $n < j$, but fail to satisfy $g_j^-(x) \leq 0$. Let $h^+_i(x) = g^+_i(x)$ for all $i \leq I$, $h^-_n(x) = g^-_n(x)$ for all $n < j$, and $h_{I+1}^{+}(x) = g_j^-(x)$. Then the set of all points to which this reason applies is $\{x \in \mathcal{B}_\epsilon(x_0) \mid \forall i \leq I+1. h^+_i(x) > 0 \wedge \forall n < j. h^-_n(x) \leq 0\}$.
\end{enumerate}

The union of these simple analytic sets is the complement of $A$.
\end{proof}

\begin{corollary}
If $A_1, \dots, A_n$ are simple analytic sets, then $\overline{\cup_{i \leq n} A_i}$ is a finite union of disjoint analytic sets.
\end{corollary}

\begin{proof}
We have $\overline{\cup A_i} = \cap\overline{A_i} = \cap(\cup_{m \leq M_i} B_{im})$, where $\{B_{im}\}_{m \leq M_i}$is a representation of $A_i$'s complement as a disjoint union of finitely many simple analytic sets. Distributing, this is in turn equal to $\cup_{m_1 \leq M_1} \dots \cup_{m_n \leq M_n}(\cap_{i \leq n} B_{im_i})$. Each element of this large union is analytic, because it is the intersection of $n$ analytic sets. Furthermore, these analytic sets are pairwise disjoint: any pair will disagree on $m_i$ for some $i$, and so the intersection will include disjoint sets $B_{im_i}$ and $B_{im_i'}$, for $m_i \neq m_i'$.  Because these two sets are disjoint, the overall intersections will be disjoint.
\end{proof}

\begin{corollary}
If $A_1, \dots, A_n$ are simple analytic sets, then $A_n \setminus \cup_{i<n}A_{i}$ is a finite union of \textit{disjoint} analytic sets, each of which is a finite intersection of \textit{simple} analytic sets. 
\end{corollary}

\begin{proof}
$A_n \setminus \cup_{i<n} A_i = A_n \cap (\overline{\cup_{i < n} A_i})$, and by the previous corollary, the right-hand term can be rewritten as a finite union of \textit{disjoint} analytic sets $\cup_{i < m}B_i$. The intersection is then equal to $\cup_{i < m}A_n \cap B_i$. Each element of this finite union is analytic because analytic sets are closed under intersection. They are also disjoint since the $B_i$ are disjoint.
\end{proof}

\begin{lemma}
\label{lem:disjoint-simple}
Let $A = \cup_{n\in\mathbb{N}} A_n$ be a countable union of simple analytic sets $A_n$. Then there are countably many \textit{pairwise disjoint} analytic sets $\{B_m\}_{m \in \mathbb{N}}$ such that $A = \cup_{m\in\mathbb{N}} B_m$.
\end{lemma}

\begin{proof}
Consider the countably many disjoint sets $A'_k = A_k \setminus (\cup_{i < k} A_i)$. Each $A'_k$ can itself be expressed as a countable union of pairwise disjoint sets, by the previous corollary.
\end{proof}

\begin{corollary}
\label{cor:disjointness}
Every c-analytic set $A$ can be written as a countable disjoint union of analytic sets.
\end{corollary}

\begin{proof}
Any open domain U can be expressed as a countable union of open balls. Therefore, in a countable union of \textit{arbitrary} analytic sets, any analytic set $A$ with a complicated open domain U can be replaced by countably many analytic sets with open-ball domains (and the same defining inequalities as A). We conclude with Lemma~\ref{lem:disjoint-simple}, using the fact that a countable family of countable families is a countable family.
\end{proof}

Next, we show a corollary that will be useful in the next section.
\begin{corollary}
\label{cor:inclispap}
An inclusion $U\subseteq V$ of c-analytic sets $U,V$ is a PAP function $U\to V$.
\end{corollary}

\begin{proof}
By Corollary \ref{cor:disjointness}, we can write $U=\bigsqcup_{i\in\NN} A_i$ and $V=\bigsqcup_{j\in\NN} B_j$ where the $A_i$ and $B_j$ are analytic sets.
Analytic sets are closed under intersection, so $\{A_i\cap B_j\}_{i\in \NN, j\in\NN}$ is a countable partition of $U$ into analytic sets. Hence, $\{(A_i\cap B_j, id)\}_{i\in \NN, j\in\NN}$ is a piecewise representation of the inclusion $U\mono V$.
\end{proof}

Finally, we show a few closure properties of c-analytic sets. 

\begin{lemma}
\label{lem:preimage-analytic}
If $f:U\to V$ is analytic and $A\subseteq V$ is analytic, then $f^{-1}(A)$ is analytic.
\end{lemma}

\begin{proof}
By definition, $A = \{x \in (\bigcap_{i} X^+_i) \cap (\bigcap_{j} X^-_j) \mid \forall i. g^+_i(x) > 0 \wedge \forall j. g^-_j(x) \leq 0\}$ for some analytic functions $g^+_i,  g^-_j $.
Thus, 
\begin{align*}
    & f^{-1}(A) \\
    &= \{x \in f^{-1}\Big((\bigcap_{i} X^+_i) \cap (\bigcap_{j} X^-_j)\Big) \mid \\
    &\qquad\quad\forall i. g^+_i(f(x)) > 0 \wedge \forall j. g^-_j(f(x)) \leq 0\} \\
    &= \{x \in (\bigcap_{i} f^{-1}(X^+_i) \cap (\bigcap_{j} f^{-1}(X^-_j)) \mid \\
    &\qquad\quad\forall i. (g^+_i \circ f)(x) > 0 \wedge \forall j. (g^-_j \circ f)(x) \leq 0\}
\end{align*}
As $f$ is analytic, it is continuous and so the sets $f^{-1}(X^+_i)$ and $f^{-1}(X^-_j)$ are open. And again, as $f$ is analytic, so are the functions $g^+_i\circ f$ and $g^-_j\circ f$.
\end{proof}

\begin{corollary}
\label{cor:preimage-c-analytic}
If $f:U\to V$ is \pap{} and $A\subseteq V$ is c-analytic, then $f^{-1}(A)$ is c-analytic.
\end{corollary}

\begin{proof}
By Corollary \ref{cor:disjointness}, we can write $A=\bigsqcup_{i\in\NN} A_i$ for analytic sets $A_i$. Thus, $f^{-1}(A)=f^{-1}(\bigsqcup_{i\in\NN} A_i)=\bigsqcup_{i\in\NN} f^{-1}(A_i)$ and it is sufficient to show that each $f^{-1}(A_i)$ is c-analytic.

$f$ is \pap{} so there exists a partition $U=\bigsqcup_{i\in\NN} B_i$ where $B_i$ are analytic sets and analytic functions $f_i:U_i\to R^n$ such that $f|_{B_i}=f_i$.
Therefore, $f^{-1}(A)=\bigsqcup_{i\in\NN} f_i^{-1}(A)\cap B_i$. 
Each $f_i^{-1}(A)$ is analytic by Lemma \ref{lem:preimage-analytic}, and so each $f_i^{-1}(A)\cap B_i$ is analytic. This means $f^{-1}(A)$ is indeed c-analytic. 
\end{proof}

\subsubsection{Category \cpap}
\label{sub:cpap}

\begin{definition}[\cpap]
We call \cpap{} the category whose objects are c-analytic sets and whose morphisms are PAP functions between them. Composition is given by the usual composition of functions.
\end{definition}

\begin{lemma}
\cpap{} is a concrete category.
\end{lemma}

\begin{proof}
\begin{itemize}
\item{\bf \cpap{} has a terminal object.} The analytic set $\mathbb{R}^0$ is terminal.
\item {\bf The functor hom$(1, -) : \cpap{} \rightarrow \Set$ is faithful.} The functor is the identity on morphisms.
\end{itemize}
\end{proof}

\begin{proposition}
\label{prop:cpap-concrete-site}
\cpap{} is a concrete site, where the coverings for a c-analytic set $U$ are given by countable c-analytic partitions of $U$, i.e., $\{(U_i)_i \mid \cup_i U_i = U \, \wedge \forall i\neq j, U_i \cap U_j = \emptyset\}$.
\end{proposition}

\begin{proof} 
First, note that our definition for the coverings is a well-defined, as inclusions are \pap{} functions by Corollary \ref{cor:inclispap}.

We now show that the given coverings satisfy the 5 axioms of a concrete site.
\begin{enumerate}
    \item Suppose we have a \cpap{} morphism $g : C \rightarrow D$, and a c-analytic partition $\{D_i\}_{i \in I}$ of $D$. Then we wish to find a partition of $C$ so that $g$ can be represented as a piecewise gluing of functions with codomains $D_i$. 
    First, note that since they are objects of \cpap{}, each $D_i$ can be further partitioned into countably many distinct analytic subsets $\{D_{ij}\}_{j \in J_i}$, $D_{ij} \subseteq D_i$. Furthermore, because $g$ is a morphism, it has a PAP representation $\{(C_k, g_k)\}_{k \in K}$ for analytic subsets $C_k \subseteq C$ and analytic functions $g_k$. For each $i, j, k$ define $C_{ijk} = \{c \in C_k \mid g_k(c) \in D_{ij}\}$: since $g_k$ is an analytic function and $D_{ij}$ is an analytic set, $C_{ijk}$ is analytic. Define $g^*_{ijk}$ to be the restriction of $g_k$ to $C_{ijk}$. Then $\{(C_{ijk}, g^*_{ijk})\}_{i \in I, j \in J_i, k \in K}$ is a piecewise representation of $g$ in which each piece's codomain is $D_i$ for some $i$.
    \item Let  $\{f_i: c_i\to c\}_{i\in I}$ be a covering of $c$. Then $\bigcup_{i\in I}Im(\conc{f_i})=\conc{c}$. This follows from the definition of \textit{partition}.
    \item The initial object $\ini$ is covered by the empty set.
    \item The identity is always covering. 
    \item Let  $\{f_i: c_i\to c\}_{i\in I}$ be a covering of $c$ and $\{g_{ij}: c_{ij}\to c_i\}_{ j\in J_i}$ be a cover of $c_i$ for each $i$. Then $\{f_i\circ g_{ij}:c_{ij}\to c\}_{i\in I,j\in J_i}\in \topo(c)$. This also follows from the definition of \textit{partition}, and the fact that a countable union of countable sets is countable.
\end{enumerate}
\end{proof}

\begin{lemma}
\label{lem:cpap-subcanonical}
\cpap{} is a subcanonical site.
\end{lemma}

\begin{proof}
Let $X$ be a representable presheaf on \cpap{}. Then, up to natural isomorphism, $X$ maps a c-analytic set $A$ to ${\bf cPAP}(A, B)$ for some fixed c-analytic set $B$. Consider a covering family $\{(A_i)\}_{i \in I}$ for $A$, and a compatible collection of plots $\phi_i \in X(A_i)$, which can be identified (via the natural isomorphism) with \cpap{} morphisms $\phi_i : A_i \rightarrow B$.  Then we must show that there is a unique $\phi \in X(A)$ whose restriction to each $A_i$ is $\phi_i$. To see this, note that each $\phi_i$ has a PAP representation $\{(A_{ij}, \phi_{ij})\}_{j \in J_i}$. By the definition of a covering family on \cpap{}, the $A_i$ are disjoint, and so $\{(A_{ij}, \phi_{ij})\}_{i \in I, j \in J_i}$ is a PAP representation; the $\phi$ we are looking for is the function it represents.
\end{proof}

\subsubsection{Admissible monos for \cpap{}}
\label{sub:admin-monos-cpap}

As an admissible set of monos, we choose those given by the monos defined on an c-analytic subset of their domain of definition.
More formally, we define $\monopap$ for \cpap{} as follows:

\begin{align*}
    \monopap(B) &=\{m:A\mono B~|~A\iso A',A'\text{ is a}\\
&\qquad\text{c-analytic subset of }B\}
\end{align*}

\begin{proposition}
\label{prop:admmonospap}
$\monopap$ is an admissible class of monos.
\end{proposition}

\begin{proof}
\begin{enumerate}
    \item For all $c\in\catC, 0\to c\in\monopap(c)$: the empty-set is c-analytic.
    \item $\monopap$ contains all isomorphisms: clear.
    \item $\monopap$ is closed under composition: clear.
    \item All pullbacks of $\monopap$-maps exist and are again in $\monopap$: this amounts to showing that the preimage $B:=f^{-1}(A)$ for a \pap{}  function $f$ and a c-analytic set $A$ is c-analytic. This is true by Corollary \ref{cor:preimage-c-analytic}.
    \item For each $c$, the function $\submono(c)\to\Set(\conc{c},\{0, 1\})$ is componentwise injective and order-reflecting, and the image of $\submono$ is closed under suprema of $\omega$-chains: the function is the identity on sets, and is thus injective and order-reflecting. Given an omega chain $\{A_i\subseteq B\}_{i\in\NN}$, we have $i<j \Rightarrow A_i\subseteq A_j$. Let $A:=\bigcup_{i\in\NN}A_i\subseteq B$. It is a countable union of c-analytic sets and is thus a c-analytic set.
    \item Given an increasing chain in $\submono(c)$, $(c_n\mono c)_{n\in\NN}$, denote its least upper bound by $\cinf \mono c$. Then the closure under precomposition (with any morphism) of the set $(c_n\mono \cinf)_{n\in\NN}$ contains a covering family of $\cinf$: as shown in the previous point, $\cinf$ is c-analytic and by Corollary \ref{cor:disjointness}, $\cinf = \bigsqcup_{i\in\NN} A_i$ where the $A_i$ are analytic sets. Note that the intersection of two analytic sets is an analytic set, and therefore the intersection of a c-analytic set with an analytic set is c-analytic.
    Thus, $(c_n\cap A_i\mono \cinf)_{n\in\NN, i\in\NN}$ is a covering family of $\cinf$ that is contained in the closure by precomposition of $(c_n\mono \cinf)_{n\in\NN}$.
\end{enumerate}
\end{proof}

\subsection{Soundness and adequacy}

We can now prove Theorem~\ref{thm:wpap-sound-adequate}:

\begin{theorem}
     \wpap{} has products, coproducts, exponentials and $\lift:\wpap{}\to\wpap{}$ is a strong lifting monad.
\end{theorem}

\begin{proof}
   It follows directly from the fact that we have seen that $\wpap{}$ is a category of $\omega$-concrete sheaves for the concrete site $\cpap{}$ (Proposition~\ref{prop:cpap-concrete-site}), and $\lift$ is obtained from an admissible class of monos on $\cpap{}$ (Propositions~\ref{prop:admmonospap} and \ref{prop:lift-is-monad}), by Proposition~7.5 in \citet{matache2022concrete}.
\end{proof}

\begin{proof}[Proof of Theorem~\ref{thm:wpap-sound-adequate}]
We apply theorem \ref{thm:wconshsound} to the concrete site $(\pap, \topopap, \monopap)$ with the admissible class of monos $\monopap$.
\end{proof}

\section{Correctness of AD}
\label{sec:correctness-ad}
The main goal of this section is to show Theorem~\ref{thm:fwd-cor-basic-wpap}. 
To do so, we apply the categorical machinery of fibrations for logical relations developed in \cite{katsumata2013relating}.
They build on the well-known theory of fibrations \cite{jacobs1999categorical}. 
As we follow their recipe closely, we recall the minimal amount of fibration theory in Section~\ref{sub:fibrations-for-log-rel} needed to understand their machinery, which we also reproduce in Section~\ref{sub:fibr-log-rel} for convenience.
We finally use the machinery of fibrations for logical relations in Section~\ref{sub:ad-correctness-fib-log-rel} to prove Theorem~\ref{thm:fwd-cor-basic-wpap}, our correctness result on AD. 

\subsection{Fibrations for logical relations}
\label{sub:fibrations-for-log-rel}

Roughly, fibrations offer a useful and unifying point of view for reasoning about generalized predicates on a category.
Indeed, the functor $p:Gl\to\catC$ from a glueing/sconing category \cite{mitchell1992notes} to the base category is an example of nice fibration. 
This justifies to study a categorical version of logical relations using fibrations. We follow closely this point of view developed in \cite{katsumata2013relating, katsumata2005semantic}.

\subsubsection{Fibrations}

We recall some basic facts from the theory of fibrations \citep{jacobs1999categorical}.

\begin{definition}[Reproduced from \cite{jacobs1999categorical}]
Let $\fib{p}{E}{B}$ be a functor.
A morphism $f:X\to Y\in\mbe$ is \textbf{Cartesian over} $u:I\to J\in\mbb$ if $pf=u$ and every $g:Z\to Y\in\mbe$ for which $pg=u\circ w$ for some $w:pZ\to I$, there is a unique $h:Z\to X\in\mbe$ above $w$ with $f\circ h=g$.
\end{definition}

\begin{definition}[Reproduced from \cite{jacobs1999categorical}]
A \textbf{fibration} is a functor $\fib{p}{E}{B}$ such that for every $Y\in\mbe$ and $u:I\to pY\in\mbb$, there is a Cartesian morphism $f:X\to Y\in\mbe$ above $u$. It is also called \textbf{fibred category} or \textbf{category (fibred) over} $\mbb$.
\end{definition}

\begin{example}
Let $\Pred$ be the category whose objects are pairs $(X,S)$ of a set $X$ and a subset $S\subseteq X$, and morphisms $(X,S)\to(Y,T)$  are functions $f:X\to Y$ such that $f(S)\subseteq T$. Let $p:\Pred\to\Set$ be the functor that forgets the subsets $S$. It is a fibration.
\end{example}

\begin{proposition}(Change-of-base, \cite{jacobs1999categorical})
Let $\fib{p}{\mbb}{\mbb}$ be a fibration and $K:\mba\to\mbb$ be a functor. Form the pullback in \textbf{Cat}:
\[
\xymatrix{
\mba\times_{\mbb}\mbe \pullbackcorner \ar[r] \ar[d]_{K^*(p)} 
& \mbe \ar[d]^{p} \\
\mba \ar[r]_{K} & \mbb
}
\]
In this situation, the functor $K^*(p)$ is also a fibration. 
\end{proposition}

\begin{example}
    Consider the functor $K:\Set\times \Set\to \Set$ that sends a pair of sets to their product. In this case, 
  \[
\xymatrix{
\BRel \pullbackcorner \ar[r] \ar[d]_{K^*(p)} 
& \Pred \ar[d]^{p} \\
\Set\times\Set \ar[r]_{K} & \Set
}
\]
$\BRel$ is a category of functions that preserve a binary relation. An object in $\BRel$ is a binary relation $S\subseteq X\times Y$, and a morphism $S\subseteq X\times Y\to R\subseteq X_2\times Y_2$ is  a function $f: X\times Y\to X_2\times Y_2$ such that for all $(x,y)\in S, f(x,y)\in R$.
\end{example}

\begin{proposition}(Composition, \cite{jacobs1999categorical})
Let $\fib{p}{\mbe}{\mbb}$ and $\fib{r}{\mbb}{\mba}$ be fibrations. Then $\fib{rp}{\mbe}{\mba}$ is also a fibration, in which $f\in\mbe$ is a Cartesian morphism iff $f$ is a Cartesian morphism for $p$ and $p(f)$ is  is a Cartesian morphism for $r$.


For each $I\in\mba$ one obtains a fibration $p_I:\mbe_I=(rp)^{-1}(I)\to \mbb_I=r^{-1}(I)$ by restriction. 
\end{proposition}

\begin{definition}[Reproduced from \cite{jacobs1999categorical}]
Given a fibration $\fib{p}{\mbe}{\mbb}$, the fibre category over $I\in \mbb$ is the subcategory $\mbe_I$ of $\mbe$ whose objects are above $I$ and morphisms above $id_I$.
\end{definition}

\begin{definition}
    $\fib{p}{\mbe}{\mbb}$ is an opfibration if $\fib{p^{op}}{\mbe{}^{\text{op}}}{\mbb{}^{\text{op}}}$ is a fibration.
\end{definition}

\begin{example}
    $p:\Pred\to \Set$ is an opfibration.
\end{example}

\subsubsection{Fibrations for logical relations for effectful languages}
\label{sub:fibr-log-rel}

We now recall the setting of fibrations for logical relations from \cite{katsumata2013relating}.

\begin{definition}[Reproduced from \cite{katsumata2013relating}]
A partial order bifibration with fibrewise small products is a faithful functor $\fib{p}{\mbe}{\mbb}$ such that
\begin{itemize}
    \item $p$ is a fibration
    \item $p$ is an opfibration
    \item each fibre category is a partial order
    \item each fibre category has small products, and the inverse image functors (necessarily) preserve them
\end{itemize}
\end{definition}

\begin{definition}
A fibration for logical relations over a bi-CCC $\mbb$ is a partial order bifibration $\fib{p}{\mbe}{\mbb}$ with fibrewise small products such that $\mbe$ is a bi-CCC and $p$ strictly preserves the bicartesian-closed structure. 
\end{definition}

\begin{proposition}[\citep{katsumata2013relating}]
\label{propn:changeofbase}
The pullback of a fibration for logical relations along a finite product preserving functor is a fibration for logical relations.
\end{proposition}

In particular, the usual subscone fibration is recovered as a change of base along the functor $\mbb(1,-):\mbb\to\Set$.

\begin{definition}[Reproduced from \cite{katsumata2013relating}]
Let $\fib{p}{\mbe}{\mbb}$ be a functor.
Given $X,Y\in\mbe$, and $f:pX\to pY$, we write $f: X\lifted Y$ to denote the following proposition: $\exists \stackrel{\cdot}{f}:X\to Y. p(\stackrel{\cdot}{f})=f$.
We say that $f$ has a lift (in $\mbe$).
\end{definition}

The setting of \citet{katsumata2013relating} works for fairly general languages with effects. One way to describe them is in terms of algebraic operations, or equivalently in terms of generic effects.

\begin{definition}[Reproduced from \cite{katsumata2013relating}]
Let $\catC$ be a category and $T$ a strong monad on it. Given objects $C,D$ f $\catC$, a $(D,C)$ algebraic operation or generic effect for is a morphism $C\to TD$.
\end{definition}

A lambda-calculus with effect is parametrized by a set of base types and (effectful) primitives, each having an arity and coarity, describing the number and types of arguments and return values. This is encapsulated as a signature: 

\begin{definition}[Reproduced from \cite{katsumata2013relating}]
A $\lambda_c$-signature $\Sigma$ is a tuple $(B,K,O,ar,car: K\bigcup O\to GType(B))$ where $B$ is a set of base types, $K$ a set of effect-free constants, $O$ a set of algebraic operations, and $GType(B)$ the set of ground types on $B$. $ar,car$ are the arity and co-arity functions.
\end{definition}

\begin{definition}[Reproduced from \cite{katsumata2013relating}]
Let $\Sigma$ be a signature. A $\lcstruct$-structure is a tuple $A=(\mbb,T,A,a)$ where $\mbb$ is a bi-CCC, $T$ a strong monad on $\mbb$, $A$ a functor $B\to\mbb$.
\end{definition}

In other words, a $\lcstruct$-structure gives an interpretation of the base types and primitives of the language.
Note that we can  always extend $A$ to a structure preserving functor $A\sem{-}:Gtype(B)\to \mbb$.

We now recall the main theorems from \cite{katsumata2013relating}:

\begin{definition}[Reproduced from \cite{katsumata2013relating}]
\label{defn:property}
    A property over a  $\lcstruct$-structure $A$ is a pair $(V,C)$ of functors $V,C:B\to \mbe$ such that $p\circ V=A$ and $p\circ C=T\circ A$
\end{definition}

That is, a property is the choice for any base type $b\in B$ of a predicate $Vb$ in the category of predicates $\mbe$ over values of type $b$, and of a predicate $Cb$ over computations $T(Ab)$ of type $b$. 
The question of interest is then whether all programs preserve the property, which is answered with the following theorem.

\begin{theorem}[Logical relations, \citep{katsumata2013relating}]
\label{thm:logrel}
Let $\Sigma$ be a signature, $A=(\mbb,T,A,a)$ be an $\lcstruct$-structure, $\fib{p}{\mbe}{\mbb}$ be a fibration for logical relations and $(V,C)$ be a property over $A$. If the property satisfies the following conditions
\begin{itemize}
    \item for all $b\in B$, $\return$ has a lift
    \item for all $k\in K$, $A\sem{k}$ has a lift
    \item for all $o\in O$, $\sem{o}$ has a lift
\end{itemize}
then for all well-typed terms $\var_1:b_1,\ldots,\var_n:b_n\vdash \ter :b$, $A\sem{\ter}$ lifts to the total category. 
\end{theorem}

\subsubsection{Fibrations for logical relations for effectful languages with recursion}

So far, we have reviewed the way to use fibrations for logical relations for higher-order effectful languages, but not for languages with recursion. We review this next.


\begin{definition}[Reproduced from \cite{katsumata2013relating}]
An \wcpo{}-enriched bi-CCC is a bi-CCC $\catC$ such that each homset is an \wcpo{}, and the composition, tupling $(-,-)$, cotupling $[-, -]$ and currying $\lambda(-)$ of the bi-CC structure on $\catC$ are all monotone and continuous. 
\end{definition}

\begin{definition}[Reproduced from \cite{katsumata2013relating}]
A pseudo-lifting monad on an \wcpo{}-enriched bi-CC category $\mbb$ is a monad on the underlying non-enriched category such that it has an algebraic operation $bt$ such that its component at $I$, $bt_I:1\to TI$, is the least morphism.
\end{definition}

\begin{definition}[Reproduced from \cite{katsumata2013relating}]
An \wcpo{}-enriched $\lcstruct$-structure is a tuple $(\mbb,T,A,a)$ such that $\mbb$ is an \wcpo{}-enriched bi-CCC, T is a pseudo-lifting monad over $\mbb$ and $(\mbb,T,A,a)$ is a $\lcstruct$-structure.
\end{definition}

\begin{definition}[Reproduced from \cite{katsumata2013relating}]
We call $X\in\mbe$ above $TI\in\mbb$ admissible if
\begin{itemize}
    \item $\bot_{pX}:1 \lifted X$
    \item for all $Y\in\mbe$ and $\omega$-chain $f_i\in\mbb(pY,pX)$ such that $f_i:Y \lifted X$, we have $\bigsqcup_{i=0}^\infty f_i: Y\lifted X$. 
\end{itemize}
\end{definition}

\begin{theorem}[Logical relations for language with recursion, \citep{katsumata2013relating}]
\label{thm:logrelwiter}
Let $\Sigma$ be a signature, $A=(\mbb,T,A,a)$ an \wcpo{}-enriched $\lcstruct$-structure, $\fib{p}{\mbe}{\mbb}$ a fibration for logical relations and $(V,C)$ be a property over $A$. If the property satisfies 
\begin{itemize}
    \item the conditions of Theorem~\ref{thm:logrel}
    \item for all $b\in\mbb, Cb$ is admissible
\end{itemize}
Then for all well-typed terms of the language with iteration,  $\var_1:b_1,\ldots,\var_n:b_n\vdash \ter :b$ we have $A\sem{\ter}:\prod_i Vb_i\lifted Cb$.
\end{theorem}

\subsection{Correctness of AD}
\label{sub:ad-correctness-fib-log-rel}

We will show correctness of AD (Theorem~\ref{thm:fwd-cor-full-wpap}) via the theory of fibrations for logical relations.
The goal is to construct a category of predicates on $\wpap{}$ and to encode correctness of AD as a predicate preservation property. 
To construct such a fibrations for logical relations, we obtain one by pulling back another one along a product preserving functor. 
There are several possibilities that give us different sort of predicates on \wpap{} spaces. 
One constraint is that we need to be able to interpret our choice of logical predicate for AD as such a predicate, and make sure all of our constants are predicate preserving. This is in particular non-trivial when dealing with constants involving coproducts, such as $<:\reals\times\reals\to\BB$ as the booleans are interpreted as the coproduct $1+1$. 
Lastly, we need to make sure that our choice of predicate for base types is closed under lubs of $\omega$-chains. 

We organize this section as follows. 
First, we construct a suitable category of predicates on $\wpap{}\times\wpap{}$ (i.e. a fibration for logical relations). 
Second, we construct a property (Definition~\ref{defn:property}) on this category of predicates.
Note that the fact that derivatives are not unique any more is what will make this logical relation more complex. 
Lastly, we will show the correctness of AD by combining these elements and Theorem~\ref{thm:logrelwiter}.

\subsubsection{A fibration for logical relations on \ensuremath{\wpap{}\times\wpap{}}}
\label{subsub:log-rel}

\begin{definition}
    We define the category $\Injpap$, a subcategory of $\cpap{}$, whose objects are c-analytic sets and morphisms are PAP injections.  
\end{definition}

\begin{lemma}
     $\Injpap$ is a site, with coverage $\mathcal{J}$ where coverings are given as in \cpap{}.
\end{lemma}

\begin{proof}
    The proof for the case of \cpap{} carries through to this restricted setting.
\end{proof}

\begin{definition}
We define $\Shinj$ to be the category of sheaves on the site $\Injpap$. 
\end{definition}

\begin{lemma}
    $\Shinj$ is a Grothendieck topos, and in particular it is Cartesian-closed, complete and co-complete.
\end{lemma}

\begin{proof}
    This is standard sheaf theory, see e.g. \citet{maclane2012sheaves}.
\end{proof}

By a subsheaf $P$ of $F$, we mean a sheaf $F$ such that for all c-analytic set $A$, $P(A)\subseteq F(A)$, and $P(f)$ is the restriction of $F(f)$ for all morphisms $f$.

\begin{definition}
    We denote by $\Subshinj$ the category whose objects are pairs of sheaves $(P,F)$ of $\Shinj$, where $P$ is a subsheaf of $F$, and morphisms $(P,F)\to (P',F')$ are natural transformations  $F\to F'$ that restrict to the subsheaves.
\end{definition}

\begin{lemma}
The second projection $(P,F)\mapsto F$ is a fibrations for logical relations $p:\Subshinj\to\Shinj$.
\end{lemma}

\begin{proof}
    As $\Shinj$ has pullbacks, it is a fibration. It is clearly faithful. 
     By the theory of glueing (\citep{carboni1995connected,johnstone2007quasitoposes,johnstone2002sketches,mitchell1992notes}), as $\Subshinj$ is a CCC and has an epi-mono factorisation system, 
      $\Subshinj$ is a CCC, has finite colimits, and $p$ preserves the CCC-structure. In addition, each fibre category is a partial order and has small products.
    It remains to show that $cod$ is an opfibration and that $p$ preserves coproducts. By Lemma 4.6 in \citet{kammar2018factorisation}, it suffices to show that monos are closed under coproducts, as we already have the other conditions for having a factorization system for logical relations. 
    The arrow category on  $\Shinj$  is cocomplete, as $\Shinj$ is (it is a topos). As $\Shinj$ is a topos, it has a epi-mono factorization system, and therefore it has image factorization. Therefore, the category of monos is  a reflective subcategory of the arrow category. As such, it is cocomplete, and in particular it has coproducts. So monos are closed under coproducts, and we conclude by Lemma 4.6 in \citet{kammar2018factorisation}, as said above.
\end{proof}

We now define the functor $F:\wpap{}\times\wpap{}\to \Subshinj$ given by $F(X,Y)= X\times \prod_{i\in \NN}Y$, where $X\times  \prod_{i\in \NN}Y$ forgets that it is an $\omega$-concrete sheaf on the site $\cpap{}$, and only remembers that it is in particular a sheaf on the site $\Injpap$. It is a product preserving functor, and therefore, by Proposition~\ref{propn:changeofbase}, the pullback of $p:\Subshinj\to\Subshinj$ along $F$ is a fibration for logical relations.
We denote this pullback by $\pi:\Gl:\to \wpap{}\times\wpap{}$.

\begin{remark}
    Before moving on, let us say a few words on this category of predicates $\Gl$. 
One reason why we defined the category of sheaves $\Subshinj$ is that the property we will define for the correctness of AD will \emph{not} be a concrete subsheaf, and this excludes directly using subobjects on $\wpap{}$ as the fibration for logical relations of interest. 
Secondly, the base type $\BB$ is interpreted as a coproduct, and if we simply choose the category of presheaves on $\Injpap$, primitives like $>$ will not lift, for a similar reason to the one presented in \citet{vakar2020denotational}. We could have chosen sheaves on $\cpap{}$ instead. This would not change much, but this is a bit overkill for our purpose and simply adds some unnecessary complexity. 
Thirdly, we do seem to need to have the property indexed by all c-analytic sets $A$.
One reason is that the fact that we have partial functions forces the definition of a lift for $\lift \RR$, as well as for $\RR$. A second reason is that the non-uniqueness of intentional derivatives creates a technical complication: being a correct intentional derivative at a point is essentially non-informative, so we cannot use the same trick as we did in the smooth case. 
That is, the problem is that the input space $\RR$ is smaller than $\RR^n$, and the quantification only shows that $h$ is `correct' on a set of at most the size of $\RR$ inside $\RR^n$, which is Lebesgue-negligible. This problem disappears when we can show that $h$ is correct on each open of an open cover of its domain.
\end{remark}

\subsubsection{A property for the correctness of AD}
\label{subsubsec:property-ad}

\begin{definition}
    Let $f:A\subseteq \RR^n\to\RR$  a PAP function. 
    Given any intentional derivative $g:A\to \RR^n$ of $f$, we call $\pi_i \circ g:A\to\RR$ an $i$-th partial intentional derivative of $f$.

    We write $hot_i$ for the vector in $\RR^n$ that consist only of zeros, except it has a single 1 at the $i$-th position.
    Given any dual-number intentional representation $g:A\times \RR^n\to \RR^2$ of $f$, define the PAP function $g_i:A\to \RR^2$ given by
    \[g_i(x)=g(x,hot_i)\]
    $g_i$ is called an $i$-th partial dual-number intentional representation of $f$, and its second component is an $i$-th partial intentional derivative of $f$. 
    We write $\partial^i_{DNR}f$ for the set of $i$-th partial dual-number intentional representations of $f$. 

    We extend the definition of $\partial^i_{DNR}f$ to all natural numbers $i$ by setting $\partial^i_{DNR}f=\{\lambda x.0\}$ for any $i> n$.
\end{definition} 

\begin{definition}
    Let $V(\RR):= (\RR,\prod_{i\in\NN}\RR\times\RR,V_\RR)$ where for each c-analytic set $A$, we have 
    \begin{align*}
        V_\RR(A)&=\{(f:A\to\RR,(g_i:A\to\RR^2)_{i\in \NN})~\mid\\
        &\qquad~f,g\text{ PAP, and }g_i\in \partial^i_{DNR}f \}
    \end{align*}

    Likewise, let $C(\RR):=(\lift \RR, \lift(\prod_{i\in\NN}\RR\times\RR),C_\RR)$ where for each c-analytic set $A$, we have 
    \begin{align*}
        C_\RR(A) &=\{(f:A\to\lift \RR,g:A\to\lift(\prod_{i\in\NN}\RR\times\RR))~\mid~\\
        &\qquad \dom(f)=\dom(g) \text{ and}\\
        &\qquad \forall i.(\widetilde{f},\widetilde{L\pi_i\circ g})\in V_\RR(\dom(f))\}
    \end{align*}

     Let $V(\BB)=(\BB,\prod_{i\in\NN} \BB,V_\BB)$ where for each c-analytic set $A$, we have 
    \[V_\BB(A)= \{f:A\to \BB, g:A\to \prod_{i\in\NN}\BB~\mid~ \forall i, \pi_i\circ g=f\}\]

    Let $C(\BB)= (\lift \BB, \lift(\prod_{i\in\NN}\BB),C_\BB)$ where for each c-analytic set $A$, we have 
    \begin{align*}
        C_\BB(A) &=\{(f:A\to\lift \BB,g:A\to\lift(\prod_{i\in\NN}\BB))~\mid~ \\
        &\qquad\dom(f)=\dom(g) \text{ and}\\
        &\qquad \forall i.(\widetilde{f},\widetilde{L\pi_i\circ g})\in V_\BB(\dom(f))\}
    \end{align*}
\end{definition}

\begin{lemma}
   We have defined a property for our language.
\end{lemma}

\begin{proof}
    The only thing to show is that $V_\RR,C_\RR,V_\BB,C_\BB$ are subsheaves of $\RR\times \prod_{i\in \NN}\RR^2, \lift \RR \times \lift\prod_{i\in \NN}\RR^2, \BB\times  \prod_{i\in \NN}\BB,  \lift \BB\times  \lift\prod_{i\in \NN}\BB$ respectively. 
    It is immediate for the $C$'s once we have shown the result for the $V$'s. For $V_\RR$, it amounts to showing two things. First, whether partial intentional derivatives restrict to a subset of the domain of a function, which is evidently true. Second, whether partial intentional derivatives glue if they agree on their intersection, which as for usual derivatives, is true as well. 
    For $V_\BB$, it is immediate as it is an equalizer of $\BB\times  \prod_{i\in \NN}\BB$, which exists in any topos and the canonical map $V_\BB\mono \BB\times  \prod_{i\in \NN}\BB$ is a mono. 
\end{proof}

Let us now look a bit more into what morphisms preserving the property look like.

\begin{definition}
    Let $V(\RR^k):=(\RR^k,(\prod_{i\in\NN}\RR\times\RR)^k, V_{\RR^k})$ where for each c-analytic set $A$, we have 
    \begin{align*}
        V_{\RR^k}(A)&=\{(f:A\to\RR^k,g:A\to(\prod_{i\in\NN}\RR\times\RR)^k)~\mid\\
        &\qquad~\forall 1\leq j\leq k, (\pi_j f,\pi_j g)\in V_\RR(A) \}
    \end{align*}

    Let $f:\RR^n\to \lift \RR$ represent a PAP function. 
     We define $V(f)$ to be pairs of functions $(f,h)$, where $h:(\prod_{i\in\NN}\RR\times\RR)^k\to \lift (\prod_{i\in\NN}\RR\times\RR)$
     is $\wpap{}$, and such that for any c-analytic set $A$, and any $(g_1,g_2)\in V_{\RR^k}(A)$, $(f\circ g_1, h\circ g_2)\in C_\RR(A)$.
\end{definition}

    This means that given a PAP functions $g_1,\ldots, g_n:A\to \RR$, $h$ will send $i$-th partial intentional representations of each of the $g_i$ to  $i$-th partial intentional representations of 
    $f \circ \langle g_1, \ldots, g_n \rangle$. 
    Note that if $A\subseteq \RR^k$, this will be trivially satisfied for any $i> k $.

    Let us now see how we will interpret and a lift primitive $\ter:\reals^2\to \reals$ of our language. This interpretation will only be required for the correctness proof.
    We set $\sem{\ter}_{\wpap{}\times\wpap{}}= (\sem{\ter}, \prod_{i\in\NN}\sem{\ad(\ter)})$. 
    In other words, we make countably many copies of the semantics of the AD-translation of $\ter$. 
    We can now more generally show that every primitive of the language, whose semantics is interpreted in $\wpap{}\times \wpap{}$ following the example above, will lift to $\Gl$.

    \begin{lemma}
\label{lem:primitive-lift}
    Every primitive of the language has a lift in $\Gl$.
\end{lemma}


Finally, to satisfy the conditions of Theorem~\ref{thm:logrelwiter}, we need to show that for each base $B$, $C(B)$ is admissible.

\begin{lemma}
  \label{lem:cb-admissible}
    \label{lem:cr-admissible}
    $C(\BB)$ and  $C(\RR)$  are admissible.
\end{lemma}

\subsubsection{Correctness of AD}

\begin{lemma}
\label{lem:wpapisbiccc}
\wpap{} is an \wcpo{}-enriched bi-CCC.
\end{lemma}

\begin{proof}
Each homset is an \wcpo{} by forgetting the \pap{}-structure.
For composition and coproducts, it is immediate as lubs are taken point-wise.
For products, if $f = \bigvee_i f_i$ and $g = \bigvee_i g_i$, then $(f,g) = ( \bigvee_i f_i, \bigvee_i g_i)=  \bigvee_i (f_i,g_i)$. 
Finally, if $f_i: X\times Y\to Z$, for currying we have $(\bigvee f_i)(x,y)=\bigvee f_i(x,y)= \bigvee \lambda(f_i)(x)(y)$. 
As lubs are taken pointwise, we have $\bigvee \lambda(f_i)(x)(y) = (\bigvee \lambda(f_i)(x))(y)$.
This means that
$\lambda (\bigvee f_i) = \bigvee \lambda(f_i)$.
The monotonicity part is similar.

In summary, it follows from the fact that the forgetful functor $\wpap{}\to \wcpocat{}$ preserves and reflects limits, coproducts, exponentials, and the \wcpo{}-structure.
\end{proof}

\begin{lemma}
\label{lem:partial-deriv-to-full-ok}
    Let $f:A\subseteq \RR^n\to \RR$ be PAP. Then $h:A\times \RR^{n}\to \RR^2$ is a dual-number representation of $f$ iff for all natural number $i$, 
    $h(x,hot_i)$ is an $i$-th partial dual number representation of $f$.
\end{lemma}


\begin{theorem}
\label{thm:log-rel-wpap}
For any $\var_1:\reals,\ldots,\var_n:\reals\vdash \ter:\reals$, 
$\sem{\ter}_{\wpap{}\times \wpap{}}$ has a lift.
\end{theorem}

\begin{proof}
We have done all the work in the previous subsections.
We simply check the condition to apply Theorem~\ref{thm:logrelwiter}.

First, note that we indeed have a semantics in $\wpap{}\times\wpap{}$, given by $\sem{\ter}_{\wpap{}\times\wpap{}}:= (\sem{\ter}, \prod_{i\in\NN} \sem{\ad(\ter)})$, extending what we have seen in Section~\ref{subsubsec:property-ad}. Then,

\begin{itemize}
    \item We have defined a fibrations for logical-relations $\pi:\Gl\to \wpap{}\times\wpap{}$ in Section~\ref{subsub:log-rel}.
    \item We have defined a property  $(V,C)$ in $\Gl$ in Section~\ref{subsubsec:property-ad}.
    \item We check that $x\mapsto\return(x)$ lifts, but this is immediate.
    \item Each primitive has a lift, by Lemma~\ref{lem:primitive-lift}.
    \item \wpap{} is an \wcpo{}-enriched bi-CCC by Lemma~\ref{lem:wpapisbiccc}.
    \item $\lift$ is a pseudo-lifting monad. It means that we have an algebraic element $\bot_A:1\to \lift A$ that is interpreted in $\lift A$ as the bottom element. This is true by construction.
    \item $C(\RR)$ and $C(\BB)$ are admissible, by Lemma~\ref{lem:cr-admissible}.
\end{itemize}
\end{proof}

\begin{corollary}[Correctness of AD (limited)]
 For any term\\ $\var_1\colon\reals,\dots,\var_n\colon\reals\vdash \trm : \reals$, the PAP function
  $\sem {\ad(\trm)}:A\times \RR^n \to \RR\times\RR$ is a dual-number intentional representation of
  $\sem{\trm}:A\subseteq \RR^n \to \RR$.
\end{corollary}

\begin{proof}
Let $\var_1\colon\reals,\dots,\var_n\colon\reals\vdash \trm : \reals$. 
By Theorem~\ref{thm:log-rel-wpap}, $\sem{\ter}_{\wpap{}\times \wpap{}}$ has a lift. Consider the c-analytic set $\RR^n$ and the pair
\[(id_{\RR^n}, \lambda x. \prod_{i\leq k\leq n} \prod_{i\in \NN}(\pi_k(x),[i=k]\pi_i'(x)))\] 
where, by convention, $\pi_i(x)=\pi_i'(x)=0$ whenever $i> n$. We use Iverson brackets, $[i=k]$, which equals to $1$ when the condition is satisfied, and $0$ otherwise.
This pair is an element of $V_{\RR^k}(\RR^n)$. This means that 
\begin{align*}
    &\Big(\sem{\ter}, \prod_{i\leq k\leq n} \prod_{i\in\NN}\sem{\ad(\ter)}\circ \big(\lambda x. \prod_{1\leq k\leq n}(\pi_k(x),\\
    &\qquad\qquad[i=k]\pi_i'(x))\big)\Big)\in C_\RR(\RR^n)
\end{align*}

This simplifies a bit more, and by the definition of $ C_\RR(\RR^n)$, this means that for all $i$, $\sem{\ad(\ter)}(x,hot_i)$ is an $i$-th intentional partial representation of $\sem{\ter}$. By Lemma~\ref{lem:partial-deriv-to-full-ok}, this means that  $\sem{\ad(\ter)}$ is a dual-number intentional representation for $\sem{\ter}$, as desired.
\end{proof}

With the same proof, Theorem~\ref{thm:log-rel-wpap} gives us directly the stronger result at all ground types:

\begin{corollary}[Correctness of AD (full)]
      For any term well-typed term $\Gamma\vdash\trm:\ty$ of ground type $\ty$ in a ground context $\Gamma$, the PAP function
  $\sem {\ad(\trm)}:A\times \RR^n \to \RR\times\RR$ is a dual-number intentional representation of
  $\sem{\trm}:A\subseteq \RR^n \to \RR$.
\end{corollary}

\begin{proposition}
\label{propn:failure-set-quasivariety}
    
\end{proposition}

\section{Failure and convergence of Gradient Descent}
\subsection{Failure of gradient descent for PAP functions}
\label{sub:failure-ad}

\begin{proof}
   Let $\epsilon>0$. 
   Let $P$ be the closed program defined by the following code, where $lr$ is our chosen learning rate $\epsilon$.
    \begin{minted}[fontsize=\small]{python}  
    # Recursive helper
    def g(x, n):
        if x > 0:
            return ((x - n) * (x - n)) / (lr * 2)
        if x == 0:
            return x / lr + n*n / (2*lr)
        else:
            return g(x+1, n+1)
    
    def P(x) = g(x,0)        
    \end{minted}
    This can easily be written in our language with recursion and conditionals, but it would slightly impact readability.  

   $\sem{P}$ is PAP, and by inspection, we see that $\sem{P}=x \mapsto \frac{x^2}{2 \epsilon}$. This is proved by simple induction on $\NN$ for the recursive function $g$ defining $P$. 
    Let us now show that $f:=\sem{P}$ satisfies the hypothesis of Theorem~\ref{thm:gd-conv} and is strictly convex:
    \begin{itemize}
        \item $f$ is strictly convex: this is well-known for the $x^2$ function
        \item $f$ is bounded below: $f\geq 0$
        \item $f$ is $L$-smooth for some $L$:
        \begin{align*}
            \grad f(x)-\grad f(y) =\frac{2x}{2\epsilon}-\frac{2y}{2\epsilon} = \frac{1}{\epsilon}(x-y) 
        \end{align*}
        And therefore $f$ is $\frac{1}{\epsilon}$-smooth.
    \end{itemize}
    This means that for all $0<\epsilon'< \frac{2}{\frac{1}{\epsilon}}=2\epsilon$, gradient descent on $f$ with rate $\epsilon'$ converges  to the global minimum of $f$, which is at $0$ at $0$. 
    Now, let us see that this is not the case when we use the intentional derivatives for $f$ computed by AD, in the case $\epsilon':=\epsilon$.
    
   Applying our standard AD macro to $P$, we obtain a PAP function $P'$ that returns $\frac{x}{ \epsilon}$ when $x\in\posreal$, and $\frac{(x+\forget{|x|}+1)-(\forget{|x|}+1)}{ \epsilon }$ when $x\in\negreal- \{-n~|~n\in\NN\}$, and $\frac{1 }{\epsilon} $ for  $x\in \{-n~|~n\in\NN\}$. Note that $\grad \sem{P}(x)\neq \sem{P'}(x)$ exactly when $x\in \{-n~|~n\in\NN\}$. 
   
   Let $x^0\in \RR$ be any initialization for the gradient descent algorithm.
   Then, after one step of gradient descent, we obtain $x^1 := x^0 - \epsilon \times P'(x^0)$. If $x^0>0$, then  $\sem{P'(x^0)}= \frac{x^0}{\epsilon}$ and thus $x^1=0$. When $x^0<0$ is not an integer, $\sem{P'(x^0}=\frac{x^0}{\epsilon}$ as thus $x^1=0$ as well.
   Finally, when $x^0<0$ is an integer, $\sem{P'(x^0)}=\frac{1}{\epsilon} $ and therefore $x^1=x^0-1$. In any case, $x^1$ is a non-positive integer, and for $x^2$ we'll always have $x^2= x^1-1$. Therefore, $x^t \to -\infty$ as $t\to \infty$.
   
    A similar program $P$ can be constructed when the learning rate is non-fixed. The recursive program inside $P$'s definition simply calls $\epsilon(n)$ instead of using $\epsilon$, and the analysis will be very similar, except that $x^1$ will not necessarily be an integer when $x^0$ is a non-positive integer, but this happens for almost no $x^0$.
\end{proof}

\subsection{Almost-sure convergence of Gradient Descent}
\label{sub:a.s-conv-gd}

\begin{proof}
    Given a fixed intentional derivative $\grad f$ of $f$, let $F_{int}(\epsilon,x)=(\epsilon,x-\epsilon\grad_xf)$.
    Let $F_{true}(\epsilon,x)=(\epsilon,x-\epsilon\grad^*_xf)$ where $\grad^*_xf$ is the true gradient of the differentiable function $f$. We restrict the domain of $F_{true}$ and $F_{int}$ to $(0,\frac{2}{L})\times \RR^d$.
    For $\delta>0$ and $t\in\NN$, we define
    \[X_{t,\delta}^*:=\{(\epsilon,x^0)\in (0,\frac{2}{L})\times \RR^d~\mid~\conc{\grad^*f(\pi_2 F^t_{true}(\epsilon,x^0))}\leq \delta\}\]
    where $g^t$ is the $t$-fold composition of a function $g$ with itself.
    As $f$ is $L$-smooth, $\grad^*f$ is locally $L$-Lipschitz and thus continuous. 
    Therefore, by the first conclusion of Theorem~\ref{thm:gd-conv}, for any $\delta>0$,
    \[\bigcup_{t=0}^\infty X_{t,\delta}^*=(0,\frac{2}{L})\times \RR^d \]
    We define similarly, for any $\delta>0$, 
     \[X_{t,\delta}:=\{(\epsilon,x^0)\in (0,\frac{2}{L})\times \RR^d~\mid~\conc{\grad^* f(\pi_2 F^t_{int}(\epsilon,x^0))}\leq \delta\}\]
    We will show by induction on $t\in\NN$ that $X_{t,\delta}$ and $X_{t,\delta}^*$ only differ by a null-set, i.e. that there exists null-sets $N_{t,\delta}$ and $M_{t,\delta}$ such that $X_{t,\delta}\cap N_{t,\delta}^c=X_{t,\delta}^*\cap M_{t,\delta}^c$.
    This implies that 
    \begin{align*}
        \bigcup_{t=0}^\infty X_{t,\delta} 
        &\supseteq \bigcup_{t=0}^\infty X_{t,\delta}\cap N_{t,\delta}\\  
        &= \bigcup_{t=0}^\infty X_{t,\delta}^* \cap M_{t,\delta} \\
        &\supseteq (\bigcup_{t=0}^\infty X_{t,\delta}^*) \cap ((0,\frac{2}{L})\times \RR^d -\bigcup_{t=0}^\infty  M_{t,\delta}) \\
        &= (0,\frac{2}{L})\times \RR^d \cap ((0,\frac{2}{L})\times \RR^d -\bigcup_{t=0}^\infty  M_{t,\delta})\\
        &= (0,\frac{2}{L})\times \RR^d -\bigcup_{t=0}^\infty  M_{t,\delta}
    \end{align*}
    As a countable union of negligible sets is negligible, this means that $\bigcup_{t=0}^\infty X_{t,\delta}$ equals $(0,\frac{2}{L})\times \RR^d$ up to a negligible set, as desired. 

    By induction on $t\in\NN$, we show the following property:
    \begin{align*}
        \forall A\subseteq (0,\frac{2}{L})\times \RR^d\text{ measurable of measure }0,\\
        \{(\epsilon,x)~\mid~ F_{true}^t(\epsilon,x)\in A\} \text{ is measurable with measure }0
    \end{align*}
    If $t=0$, $F_{true}^t$ is the identity and it's obvious.
    Let $t\in\NN$ such that the property holds. Let $A\subseteq (0,\frac{2}{L})\times \RR^d$ be of measure 0.
    To use the induction hypothesis, it suffices to show that $F_{true}^{-1}(A)$ has measure 0. 
    The Jacobian matrix of $F_{true}$ is an $(n+1)\times (n+1)$ matrix, given by
    \[J(F_{true})_{\epsilon,x} = \begin{pmatrix} 1 & 0_n^T \\ J_xf^T & I_n-\epsilon H_xf&  \end{pmatrix}\] 
    where $I_n$ is the $n\times n$ identity matrix, $0_n$ the zero vector in $\RR^n$, $H_xf$ is the Hessian matrix of $f$, and $(-)^T$ is the transpose operation.
    
    The rank of $J(F_{true})_{\epsilon,x}$ is lower bounded by one plus the rank of $I_n-\epsilon Hf_x$. The latter is strictly less than $n$ exactly when $\frac{1}{\epsilon}$ is an eigenvalue of $H_xf$. For a fixed $x$, this can only occur for finitely many $\epsilon$ (at most $n$). 
    Thus, $J(F_{true})_{\epsilon,x}$ has max rank $n+1$ almost everywhere. Denote by $R$ this set.
    $f$ being PAP, $F$ is PAP as well. Therefore, $F_{true}$ is almost everywhere analytic, and denote this set $C$.
    Let $N=R\cap C$. Note that $N^c$ is a null-set.
    On $N$, $F$ is $F\in\mathcal{C}^1$ and its Jacobian is invertible.
    Therefore, by the implicit function theorem, the result is locally true around any point $x\in N$. This means, for all $x\in N$, there exists a neighborhood $U_x$ of $x$ such that $F_{true}^{-1}(A)\cap U_x$ is a null-set. We can use a countable cover $\{U_i\}_{i\in\NN}$ of $N$ of such neighborhoods, and thus 
    \begin{align*}
        F_{true}^{-1}(A)
        &= F_{true}^{-1}(A)\cap (N\cup N^c) \\
        &\subseteq F_{true}^{-1}(A)\cap \big[(\bigcup_{i\in\NN} U_i) \cup N^c\big] \\
        &\subseteq \big[\bigcup_i F_{true}^{-1}(A)\cap U_i\big] \cup Z
    \end{align*}
    $F_{true}^{-1}(A)$ is measurable as $A$ is, and as null-sets are closed under countable unions, this shows that $F_{true}^{-1}(A)$ has measure 0, which concludes the induction.

    Now, we show by induction on $t\in\NN$ that for almost all $(\epsilon,x^0)\in (0,\frac{2}{L})\times \RR^d$,
    \[F_{true}^t(\epsilon,x^0)=F_{int}^t(\epsilon,x^0)\]
    Once again, the base case is immediate.
    Let $t\in\NN$ such that the property holds.
    As $\grad f$ is an intentional derivative of $f$, for almost all $x\in \RR$, we have that $\grad_xf=\grad_x^*f$.
    Therefore, for almost all $(\epsilon,x)$, we have that $F_{true}(\epsilon,x)=F_{int}(\epsilon,x)$. Let $A$ be the set of points on which they differ. By the previous proposition, $F_{true}^{-1}(A)$ has measure 0.
    Therefore, $F^{t+1}_{true}$ and $F^{t+1}_{int}$ may differ only on the set $F_{true}^{-1}(A)$, which has measure 0. This concludes the induction.
    
    
    As a corollary of what we have just shown, we have that for almost all $(\epsilon,x^0)$, $x^t$ is the same, whether it was obtained from true or AD-computed gradients. Indeed, $x^t$ is either $\pi_2F^t_{true}(\epsilon,x^0)$ or $\pi_2F^t_{int}(\epsilon,x^0)$, and we have shown that $F^t_{true}(\epsilon,x^0)$ and $F^t_{true}(\epsilon,x^0)$ agree almost everywhere.
        Therefore $f(x^{t=1})\leq f(x^t)$ remains true almost-everywhere for all $t$, as a countable union of negligible sets is negligible.  
        Finally, if $f$ is strongly convex, then $\{x^t\}_t$ will reach the global minimum. As $x^t$ is the same, whether it was obtained from true or AD-computed gradients, for almost all $(\epsilon,x^0)$, this means that gradient-descent with AD-computed gradients will also almost surely converge to the global minimum.
\end{proof}
\section{s-Hausdorff measures}
\subsection{Monad of measure on \wpap{}}
\label{sub:monad-of-measures}

The goal of this section is to give the key steps in the proof of Theorem~\ref{thm:trace-semantics-is-pap}. 
These are the following.
We adapt the proof technique used in \citet{vakar2019domain}.
We define an expectation operator $\mathbf{Int}_X:\mathbf{S}X\to (X\to \mbw) \to \mbw)$. 
We show that it is a well-defined $\wpap{}$-morphism, and natural in $X$.
Next, we show that $\wpap{}$ admits a proper orthogonal-factorization system, extending the known result on categories of concrete sheaves to this category of $\omega$-concrete sheaves. 
Using a theorem from \citet{kammar2018factorisation}, we show that the image of $\mathbf{Int}$ induces a strong monad $\mathbf{M}$ on \wpap{}.
This monad is easily shown to be commutative. 

Using this, we can show Lemma~\ref{thm:trace-semantics-is-pap}.

\begin{proof}
    $\sem{\ter}_\mathbf{M}$ in the new measure semantics is given by $\mathbf{Int}_{\RR^n} \circ \sem{\ter}_\mathbf{S}$. By definition, this means it is the integrator $\lambda f.\int f(\alpha_{\sem{\ter}_\mathbf{S}}(x))\Lambda_\Omega (dx)$, which can be seen as the pushforward of $\Lambda_\Omega$ by the function $\alpha_{\sem{\ter}_\mathbf{S}}$, defined by the second point in Definition~\ref{defn:integrator_sampler}. This function is an $\wpap{}$ morphism $\Omega \to \mathbf{L}(\RR^n)$, i.e. a partial \wpap{} function $f_\ter:\Omega \rightharpoonup \RR^n$ such that $\sem{\ter} = f_{\ter *}\Lambda_\Omega$.
\end{proof}

\subsection{Pushforward of PAP functions}
\label{sub:pushforward-pap}

In this section, we prove the following Theorem:
\begin{theorem}
\label{thm:shausdorff}
    Let $f : U\to \RR^n$ be the total restriction of a partial \wpap{} morphism to its domain $U \subseteq \RR^m$. Let $\lambda$ be the restriction of the Lebesgue measure to $U$. Then $f_*\lambda$ has a density w.r.t. some s-Hausdorff measure $\mu$ on $\RR^n$.
\end{theorem}

It is the core technical theorem that is required to prove Theorem~\ref{cor:main}.
To prove Theorem~\ref{thm:shausdorff}, we will first prove a restricted form of it:

\begin{theorem}
\label{thm:manifold_analytic_open}
Let $f:U\subseteq \RR^n\to \RR^m$ be real analytic on a connected open $U$. 
Then there is an integer $k$ such that $f(U)$ is contained is a countable union $M:=\bigcup_{i\in\NN}M_i$ of smooth submanifolds $M_i$ of $\RR^m$ of dimension $k$, up to a set of $k$-Hausdorff measure 0. 
Furthermore, $f_*\lambda$ has a density w.r.t. the $k$-Hausdorff measure on $M$.
\end{theorem}

Before proving the theorem, we need a few other technical lemmas.

\begin{lemma}
\label{lem:abs-cont-wrt-base-measure}
Let $\mu = \sum_{0\leq k\leq m}\sum_{i\in\NN}f_i^k.\mu^k_i\in \mathcal{M}\RR^m$ for some non-negative measurable functions $f_i^k:\RR^m\to \RR^+\cup\{\infty\}$, and such that $\mu^k_i$ is the $k$-Hausdorff measure on some $k$-dimensional submanifold $M_i^k$ of $\RR^m$.
Let $B$ be the s-Hausdorff measure given by the $(\mu^k_i)$. Then $\mu$ has a density w.r.t. $B$.
\end{lemma}

\begin{proof}
Let $X_k = \bigcup_i M_i^k$ and $\mu_k$ be the $k$-Hausdorff measure on $X_k$.
Let $g_k =\sum_{i\in\NN} 1_{x\in M^k_i}.f_i^k$. As the $f_i^k$ are non-negative and can take the value $\infty$, the limit exists and therefore the $g_k$ are measurable.
By construction, we have the equality $\sum_{i\in\NN}f_i^k.\mu^k_i= g_k.\mu^k$ for all $0\leq k\leq m$.
This follows from the simpler equality $f_1.\nu|_A+f_2.\nu|_B=(f_1+f_2).\nu|_{A\cup B}$ which is also valid for countable unions, given that the $\mu^k_i$ are all restrictions of the $k$-Hausdorff measure.
Therefore $\mu = \sum_{1\leq k\leq m}g_k.\mu_k$.

Let $g=g_0+1_{x\not\in X_0}(g_1+1_{x\not\in X_1}(g_2+1_{x\not\in X_2}(\ldots + 1_{x\not\in X_{m-1}}g_m)))$. $g$ is clearly measurable and non-negative (with the convention that $0.\infty = 0$ here).
For all $i<k$, $X_i$ is $H^k$ negligible. Thus $g_k.\mu_k=g_k.\mu_k|_{X_k-X_i}= (1_{x\not\in X_i}g_k).\mu_k$. 
From this fact, we deduce that $\sum_{1\leq k\leq m}g_k.\mu_k = g.B$, and therefore conclude that $\mu = g.B$, i.e. that $\mu$ has a density w.r.t. $B$.
\end{proof}

\begin{lemma}
\label{lem:analytic-max-rank-dense}
Let $f:U\subseteq \RR^n\to \RR^m$ be real analytic on a connected open $U$. 
Let $k$ be the maximal rank of $f$ on $U$. 
Then $V:=\{x\in U~|~rank(f(x))=k\}$ is 
\begin{itemize}
    \item open,
    \item dense in $U$,
    \item and $\lambda(U-V)=0$.
\end{itemize}
\end{lemma}

\begin{proof}
Let's first show $V$ is open. 
Given a coordinate system for $U$ and $V$, for any $x\in V$, denote by $G(x)$ the matrix for $T_xf$ in these bases. By hypothesis, $G(x)$ has at least an invertible $k\times k$ sub-matrix, and this is a characterisation of rank. Then denote by $H:M_{n\times m}(\RR)\to\RR:A\mapsto \sum_{B~k\times k~submatrix~of~A}|det(B)|$. $H$ is continuous as the determinant is, and $V= (HG)^{-1}(\RR-\{0\})$ is thus open.
Let's show $V$ is dense. 
Note that this is not true in general for smooth functions. Pick a base for $U$ and $V$ again, and let $G(x)$ the matrix for $T_xf$ in these bases. Thus $x\mapsto G(x)$ is also real analytic.

Let $p\in \bar{V}-V$. Pick a $k\times k$ submatrix $B$ of matrices in $M_{n\times m}(\RR)$, then let $H_B:M_{n\times m}(\RR)\to\RR$ be the determinant of that submatrix. 
Assume $H_BG$ is $0$ on a open neighbourhood $W_B$ of $p$ for all $B$. Then $W = \cap W_B$ is a finite intersection of opens and thus is open, and $W\cap V\neq \emptyset$ as by assumption $p$ is in the boundary of $V$. But that's a contradiction as we know that for some $B$ and any $x$ in the intersection, $H_BG(x)\neq 0$. 
Therefore, there is a $B$ such that $H_BG$ is not 0 at some point $x$ in a neighbourhood $W_B$ on $p$. 
As $f$ is analytic, so is $df$, and so is $G$.
The determinant is a polynomial function of its arguments, it is thus real analytic, and thus so is $H_B$.
As $H_BG$ is analytic and non-identically 0 on $W_B$, it must be non-zero on a dense subset of $W_B$. 
Let's assume that the closure $\bar{V}$ is not all of $U$. We obtain a contradiction by picking a point close to $\bar{V}$ that would end up in one of the $W_B$. 
In more detail, assume by ways of contradiction that $U-\bar{V}\neq \emptyset$. Let $p\in U-\bar{V}$. Let $g:U\to\RR^+$ be defined as $g(x)=d(x,\bar{V})$. By assumption, as $\bar{V}$ is closed, we have $g(p)>0$. Also, note that $g$ is continuous. As $k$ is the maximal ranked reached by $f$, $\bar{V}\neq\emptyset$. Thus, there exists $y\in\bar{V}$ such that $g(y)=0$. As $U$ is connected (and thus path connected, as an open of a Euclidean space), there is a
continuous path $c:[0,1]\to U$ such that $c(0)=y$ and $c(1)=p$.
The composition $g\circ c$ is continuous and so there a point $x\in\bar{V}-V$ such that $c(t)=x$ and $g(x)=0$ for some $t_0<1$. By compactness of $[0,1]$, let's choose the biggest such $t_0$, which exists as $c(1)=p\not\in \bar{V}$. For all $t>t_0$, we therefore have $g(c(t))>0$. However, by the result above, there is a dense open subset $W$ of a $W_B$, a neighbourhood of $x$, such that $W\subseteq V$. By density of $W$, this means that $c(t)=0$ on $c^{-1}(W_B)\cap[t_0,1]\neq\emptyset$, a contradiction.

We will now prove that $X:=U-V$ has Lebesgue measure 0.
For a fixed $k\times k$ matrix $B$ of any $n\times m$ matrix, consider the same analytic function $H_B$ as above.
Then $g_B:= G;H_B:U\to \RR$ is analytic. Using the theorem from \citet{mityagin2015zero}, it is either the 0 function or the preimage of $0$ has Lebesgue measure 0.
We can write $X = \bigcap_{B~k\times k~submatrix}g_B^{-1}(\{0\})$.
As $f$ has rank $k$ somewhere, not every $g_B$ is the 0 analytic function, and thus $X$ is contained in a Lebesgue measure 0 set, and therefore has Lebesgue measure 0.
\end{proof}

\begin{lemma}
\label{lem:pushforward-negligible}
Let $f:A\to C$ where $A\subseteq \RR^n$ is $\lambda$-measurable, and let $U\subseteq A$ be $\lambda$-measurable.
Assume that $\lambda(A-U)=0$. Then $f_*\lambda = (f|_U)_*\lambda$.
\end{lemma}

\begin{proof}
\begin{align*}
    f_*\lambda(S)
    &= \lambda(f^{-1}(S))\\
    &= \lambda(f^{-1}(S)\cap A)\\
    &= \lambda(f^{-1}(S)\cap (A-U+U))\\
    &= \lambda((f^{-1}(S)\cap (A-U))\cup (f^{-1}(S)\cap U)) \\
    &\leq   \lambda(f^{-1}(S)\cap (A-U)) + \lambda(f^{-1}(S)\cap U)\\
    & \leq  \lambda (A-U) + \lambda|_U(f^{-1}(S) \\
    &= 0+ \lambda|_U(f^{-1}(S)) \\
    &= (f|_U)_*\lambda(S)
\end{align*}
As the other inequality is obvious, we showed that $f_*\lambda(S) = (f|_U)_*\lambda(S)$. 
\end{proof}

\begin{lemma}
\label{lem:image-is-union-manifold}
Let $f:U\subseteq \RR^n\to \RR^m$ be real analytic on a connected open $U$.
Further, assume $f$ has constant rank. Then $f(U)$ is a countable union of manifolds.
\end{lemma}

\begin{proof}
As $f$ has constant rank, by the constant rank theorem, $f$ is locally of the form $(x_1,\ldots, x_n) \mapsto (x_1,\ldots,x_k,0,\ldots,0)$. 
More precisely, around any $x\in U$,  there is a neighbourhood $W$ of $x$ and smooth diffeomorphisms $\phi:\RR^m\to \RR^m,\psi:\RR^n\to\RR^n$ such that $g(x_1,\ldots, x_n) := \phi \circ f|_W\circ \psi(x_1,\ldots, x_n) = (x_1,\ldots,x_k,0,\ldots,0)$.

If $k=m$, $\RR^k=\phi \circ f|_W\circ \psi(\RR^n)$ is a $k$-dimensional smooth manifold. 
If $k<m$, let $\pi: \RR^m \to \RR^{m-k}$ given by $\pi(x_1,\ldots,x_m)=(x_{k+1},\ldots,x_m)$. $\pi$ is an orthogonal projection and therefore a submersion. By the submersion theorem, $\pi^{-1}(0)=\phi \circ f|_W\circ \psi(\RR^n)$ is a $k$-dimensional smooth manifold.
Therefore, as $\phi$ is a diffeomorphism, $f|_W(W)=f|_W\circ \psi (\RR^n)=\phi^{-1}(\pi^{-1}(0))$ is also a $k$-dimensional smooth manifold.

Now consider the open cover of $U$ given by $\{W_x~|~x\in U\}$ where $W_x$ is a neighbourhood of $x$ as above.
As opens of Euclidean spaces are Lindelöf spaces, we can extract a countable open sub-cover $\{W_{x_n}~\mid~x_n\in U\}_{n\in\NN}$ of $U$. From what we just proved, for each $n$, $f(W_{x_n})$ is a $k$-dimensional manifold.
Thus, $F(V)=\bigcup_{n\in\NN} f(W_{x_n})$ is a countable union of $k$-dimensional manifolds. 
\end{proof}

\begin{proof}[Proof of Theorem~\ref{thm:manifold_analytic_open}]
Let $k$ be the maximal rank of $f$ on $U$. 
Let $V:=\{x\in U~|~rank(f(x))=k\}$.
Let's first show that it is sufficient to show the result for $f|_V$.

By Lemma~\ref{lem:analytic-max-rank-dense}, $V$ is open and dense in $U$, and $U-V$ has Lebesgue measure 0. 
By Lemma~\ref{lem:pushforward-negligible}, $f_*\lambda = (f|_V)_*\lambda$, so we reduced the second statement to showing it for $(f|_V)$.
In addition, by Sard's theorem, $f(U-V)$ has $k$-Hausdorff measure 0. 
Therefore, it is sufficient to prove the theorem for $f|_V$, which we call $f$ again. 

Now, write $V=\bigsqcup_i V_i$ where the $V_i$ are the connected components of $V$. 
On each connected component $V_i$,  $f|_{V_i}$ verifies the hypothesis of Lemma~\ref{lem:image-is-union-manifold}.
Thus $f(V_i)$ is a countable union of $k$-dimensional manifolds for each $V_i$. As Euclidean spaces are locally connected spaces with a countable basis, the same holds for $V$. Therefore, $V$ has at most countably many connected components $V_i$. A countable union of countable sets is countable, and we finished proving the first part of the theorem.

Finally, it remains to show that $f_*\lambda$ has a density w.r.t. the $k$-Hausdorff measure on $f(V)$. 
By Lemma~\ref{lem:abs-cont-wrt-base-measure}, it is sufficient to show it for each $(f|_W)_*\lambda, f(W)$, where $W$ is as constructed above. 
First, let's show that it is also sufficient to restrict to the case where $W$ is contained in a compact set. There exists a countable cover $K_i$ of compact sets of $\RR^n$ that cover $W$ (e.g. by taking closed spheres of radius 1 at points with rational coordinates). Each $W_i:= K_i\cap W$ is thus compact in $W$. Assume for now that we have proved that $(f|_{W_i})_*\lambda$ has a density w.r.t. $f(W)$, i.e. that $(f|_{W_i})_*\lambda(A) = \int_A h_i(x) d\hausd^k(x)$ for some measurable function $h_i:f(W) \to \RR^+\cup\{\infty\}$. Then, 
\begin{align*}
    f_*\lambda(A) &=\sum_{i\in \NN} (f|_{W_i})_*\lambda(A) \\
    &=\sum_{i\in \NN}\int_A h_i(x) d\hausd^k(x) \\
    &= \int_A (\sum_{i\in \NN} h_i(x)) d\hausd^k(x)
\end{align*}
where the last equality holds by the Fubini-Tonelli theorem as $h_i\geq 0$. $h:=\sum_i h$ is measurable, and we are done.

So we can now assume that $W$ is contained in a compact set. As $f$ is analytic, it is locally Lipschitz and is thus Lipschitz on $W$, as $W$ is contained in a compact set.
What's more, As $f$ has rank $k$, its Jacobian matrix $J_x$ at every point $x$ has rank $k$. A standard result in linear algebra asserts that $J_xJ_x^T$ also has rank $k$ and therefore the correction term $J_Wf(x)$ is non-zero for all $x$. Let $g(x):=\frac{1}{J_Wf(x)}$ is therefore measurable and non-negative. By the co-area formula, we have
\begin{align*}
    f_*\lambda(A) &= \lambda (f^{-1}(A)) \\
    &= \int_{f^{-1}(A)}1.d\lambda \\
    &= \int_{f^{-1}(A)}g(x)J_Wf(x).d\lambda \\
    &= \int_{f(W)} \Big(\int_{f^{-1}(y)}g(x)d\hausd^{n-k}(x)\Big) d\hausd^k(y) \\
    &= \int_{f(W)} h(y) d\hausd^k(y)
\end{align*}
where $h(y):= \int_{f^{-1}(y)}g(x)d\hausd^{n-k}(x)$ is measurable and non-negative, yet possibly infinite.

\end{proof}

\begin{lemma}
\label{lem:big_open}
Let $A$ be an analytic set. There is an open $U\subseteq A$ such that $\lambda(A-U)=0$.
\end{lemma}

\begin{proof}
Let $A= \cap_i X_i^+ \cap \cap_j X_j^-$ as in the definition of analytic sets.
Write $X_j^-=X_j \sqcup Z_j$ where $Z_j = (g_j^-)^{-1}(\{0\})$.
Without loss of generality, assume that none of the $g_j^-$ are the constant 0 function.
Then, $U:= \cap_i X_i^+ \cap \cap_j X_j$ is clearly open and $U\subseteq A$.
Let's now show that $A-U$ has Lebesgue measure 0. 
It suffices to show that this is the case for each $Z_j$, as they cover $A-U$.
However, by \citet{mityagin2015zero}'s result,  $g_j^-$ is either the 0 function or the preimage of $0$ has Lebesgue measure 0.
This shows that the $Z_j$ have Lebesgue measure $0$, as desired.
\end{proof}

\begin{theorem}
\label{thm:pap_pushforward}
Let $f:A\to \RR^m$ be \pap{}. Then $f_*\lambda$ has a density w.r.t. an s-Hausdorff measure $\mu$ on $\RR^m$.
\end{theorem}

\begin{proof}
Let $f:A\to \RR^m$ be \pap{}. $A$ is c-analytic, so we can write it as $A=\bigsqcup_{i\in \NN} A_i$ where each $A_i$ is analytic and $(f_i,A_i)$ is a PAP representation for $f$.
By Lemma~\ref{lem:big_open}, for each $i\in\NN$, there is an open $U_i\subseteq A_i$ such that $\lambda(A_i-U_i)=0$. 
By Lemma~\ref{lem:pushforward-negligible}, it is thus sufficient to consider the case where $f_i$ is defined on an open set $U_i$.
Any open $V\subseteq\RR^n$ can always be written as $V=\bigsqcup_{i\in \NN} V_i$ for a countable family of connected opens $V_i$.
Applying this to the $U_i$, we have a countable family of analytic functions $f_{i,j}:U_{i,j}\to \RR^m$ defined on connected opens $U_{i,j}$, where $f_{i,j}$ is the restriction of $f_i$ to $U_{i,j}$, and $j\in\NN$.
By theorem~\ref{thm:manifold_analytic_open}, for each $(i,j)\in \NN^2$, there exists an integer $0\leq k_{i,j}\leq m$ such that $(f_{i,j})_*\lambda = \sum_{l\in \NN} f_l^{k_{i,j}}m^{k_{i,j}}_{l}$, where each $f_l^{k_{i,j}}:\RR^m\to \RR^+\cup\{\infty\}$ is measurable, and each $m^{k_{i,j}}_{l}$ is the $k_{i,j}$-Hausdorff measure on some $k_{i,j}$-submanifold $M^{i,j}_l$ of $\RR^m$. 
Therefore, 
\begin{align*}
    f_*\lambda &= \sum_{i,j} (f_{i,j})_*\lambda \\
    &= \sum_{i,j}\sum_{l\in \NN} f_l^{k_{i,j}}m^{k_{i,j}}_{l} \\
    &=  \sum_{0\leq k\leq m}\sum_{\{(i,j)~|~k_{i,j}=k\}}\sum_lf_l^{k_{i,j}}m^{k_{i,j}}_{l} \\
    &= \sum_{0\leq k\leq m}\sum_nf_n^{k}m^{k}_n
\end{align*}
where in the last equality we simply re-index the countable sum.
We can now conclude by Lemma~\ref{lem:abs-cont-wrt-base-measure} that $f_*\lambda$ has a density w.r.t. some s-Hausdorff measure.
\end{proof}


%
Finally, we can prove Theorem~\ref{cor:main}.

\begin{proof}
Let $\vdash \ter : \reals^n$. 
By Lemma~\ref{thm:trace-semantics-is-pap}, there exists an \wpap{} morphism $f:\Omega\to \mathbf{L}\RR^n$ such that $f_*\Lambda_\Omega= \sem{\ter}$. As $\Omega$ is a coproduct, by the universal property of the coproduct, $f$ decomposes as a family of $\wpap{}$ morphisms $f_i: \RR^i \to \mathbf{L}\RR^n$. By definition of the lifting monads, the $f_i$ represent PAP functions $f_i:A_i\to \RR^n$ where each $A_i$ is c-analytic. By Theorem~\ref{thm:pap_pushforward}, each $f_{i*}\lambda$ has a density w.r.t. an s-Hausdorff measure $\mu_i$ on $\RR^n$, i.e. $f_{i*}\lambda=g_i.\mu_i$ for some non-negative measurable function $g_i:\RR^n\to \RR^+\cup \{\infty\}$. Therefore
\[ \sem{\ter} = f_*\Lambda_\Omega 
    = \sum_{i\in\NN} f_{i*}\lambda 
    = \sum_{i\in\NN} g_i.\mu_i \]
And we conclude that this last formula has a density w.r.t. some base measure on $\RR^n$ by Lemma~\ref{lem:abs-cont-wrt-base-measure}. We proved the first point in the theorem.

For the second point, let $g$ be the density of $\llbracket t \rrbracket$ w.r.t.  $\mu$, and $f$ the density of $\llbracket t \rrbracket$ w.r.t. $B(t)$.
Write $B(t) = \sum_d H^d(- \cap M^d)$.
    
    By induction on $d$ in $[0,n]$, we show that $g(x) = f(x)$ for all $x \in M^d$, except for a $H^d$-null set.

    \begin{itemize}
        \item Base  case $d=0$: Let  $x \in M^0$.
        $\llbracket t \rrbracket(x) = (g.\mu)({x})= g(x)$$\llbracket t \rrbracket({x}) = (f.B(t))({x}) = f(x)$
        \item Inductive case $n \geq d>0$.
        \begin{enumerate}
            \item For all  $i<d$, $M^i$ is a $H^d$-null set, and therefore it suffices to show the result for $S^d := M^d-(\bigcup_{i<d} M^i)$.
            \item For any $H^d$ measurable $A \subseteq M^d$, we have:
    $\llbracket t \rrbracket(A) = g.\mu(A) = \int_{A\cap S^d}g(x) H^d(dx)$
        
        $\llbracket t \rrbracket(A) = f.B(t)(A) = \int_{A\cap S^d}f(x) H^d(dx)$
        Therefore the integrators/measures $\int_{-}g(x) H^d(dx)$ are equal as measures on $S^d$, and therefore $f(x) = g(x)$ for all $x$ on $S^d$ except on a $H^d$ null set.
        \end{enumerate}
        
        Let $Z = (\bigcup_d M^d)^c$. Then$\llbracket t \rrbracket(Z) = f.B(t) (Z) = 0.$
        Let $C = \bigcup_d C_d$ where $C_d$ is the $H^d$-null set in $S^d$ on which $f$ and $g$ disagree. Then $\llbracket t \rrbracket(C) = 0$ and thus $\llbracket t \rrbracket(Z\cup C )= 0$. As the set of points on which $f$ and $g$ disagree is contained in $Z\cup C$, we are done.
    \end{itemize}

For the third point, by definition of base measures, $B(t) = \sum_d H^d(- \cap M^d)$. By the way we constructed $B(t)$, the density $\rho_t$ satisfies $\rho_t(x)>0$ for every x in $M^d$, except at an $H^d$-null set. Let $A$ be such that $B(t_2)(A)=0$. Then, by the positive density $\rho_{t_2}$ of $\llbracket t_2 \rrbracket$ w.r.t. the base measure $B(t_2)$, we also have $\llbracket t_2 \rrbracket(A)=0$. Therefore, $\llbracket t_1 \rrbracket(A)=0$ using the assumption of absolute continuity. Finally, using the first point in Theorem~\ref{cor:main}, we conclude that $B(t_1)(A)=0$.

For the fourth point, this follows from general properties of s-Hausdorff measures. Let $\mu_1 = \sum_d H^d(- \cap M^d)$ and $\mu_2 = \sum_d H^d(- \cap M’^d)$ be two s-Hausdorff measure such that $\mu_1 \ll \mu_2$. As they are $\sigma$-finite measures, by the Radon-Nikodym theorem we can write $\mu_1 = g. \mu_2$. Let $S^d := M^d\setminus \bigcup_{i<d}(M'^i \cup M^i)$ and $A:= \bigcup_d S^d$. Let $B\subseteq A\cap M^d$. We have $\mu_1(B) = H^d(B\cap M^d)$, as for each $i<d, H^i(B\cap M^i)\leq H^i(A\cap M^d \cap M^i)=H^i(\emptyset)=0$, and for each $i>d, H^i(B\cap M^i) \leq H^i(M^d) =0$. Likewise, $\mu_2(B) = H^d(B\cap M'^d)$. As $\mu_1(B)=g.\mu_2(B)$, we have proved that the measures $H^d(-\cap M^d)$ and $g.H^d(-\cap M'^d\cap A)$ are the same measures when restricted to $A\cap M^d$, and therefore $g(x) =1$ for  $H^d$ almost all $x\in M^d\cap A.$  As $A=\bigcup_d A\cap M^d$  is measurable and $\mu_1(A^c)=0$, we conclude that the \wpap{} morphism $1_A: A\to W$ is a density for $\mu_1$ w.r.t. $\mu_2$.
\end{proof}

To conclude, we prove Theorem~\ref{thm:definability-s-hausdorff}.

\begin{proof}
By a theorem of Whitney (refined by Morrey and Grauert), all (finite-dimensional, second-countable, Hausdorff, without boundary) smooth manifolds admit a (unique) analytic structure, and two smooth manifolds are diffeomorphic iff they are analytic-equivalent. 
Our case is even simpler, as we consider embedded submanifolds of $\RR^n$. Let $M$ be a smooth submanifold of $\RR^n$ of dimension $d$. There exists a tubular envelope $U\subseteq \RR^n$ of $M$, and a smooth projection $p:U\to \RR^n$ such that $M=p(U)$. This means that  $p_*\lambda|_{U}$ and $H^d$ on $M$ are mutually absolutely continuous. 
By reasoning in local coordinates, using analytic charts, we see that $p$ is locally a linear projection, and is thus analytic, and therefore PAP.
Now let $\mu$ be an arbitrary s-Hausdorff measure on $\RR^n$, and $\{M^k_i\}_{i\in \NN,0\leq k \leq n}$ be a support for $\mu$. 
Let $p_{i,k}:U_{i,k}\to \RR^n$ be analytic projections such that $p(U_{i,k})=M^k_i$, where $i\in\NN,0\leq k\leq n$. 
Every open set $U\subseteq \RR^n$ is diffeomorphic to a bounded open $V$ of radius at most $1$, and the diffeomorphism and its inverse can be chosen to be analytic functions. Let $\phi_{i,k}:V_{i,k}\to U_{i,k}$ be such diffeomorphisms.
Let $(W_{i,k})_{i\in\NN,0\leq k\leq n}$ be the same opens as the $(V_{i,k})_{i\in\NN,0\leq k\leq n}$, but translated to be disjoint, which is always possible as the $(V_{i,k})_{i\in\NN,0\leq k\leq n}$ have radius at most $1$ and there are countably many of them. Let $t_{i,k}:W_{i,k}\to V_{i,k}$ be the associated translations, which are linear and therefore analytic.
Let $W=\bigsqcup_{i\in \NN,0\leq k\leq n} W_{i,k}$. Then the function $f:W\to \RR^n$ defined by $x\in W_{i,j} \mapsto t_{i,k};\phi_{i,k};p_{i,k}$ is analytic, and $f_*\lambda$ and $\mu$ are mutually absolutely continuous.
\end{proof}

\end{document}